\documentclass[10pt]{article}

\usepackage{amsmath,amssymb,amsfonts, amsthm}
\usepackage{mathtools}
\usepackage{verbatim}
\usepackage{boxedminipage}
\usepackage{bbold}
\usepackage{epsfig}
\usepackage{latexsym,nicefrac,bbm}
\usepackage{xspace}
\usepackage{color,fancybox,graphicx,url,subfigure}
\usepackage{fullpage}
\usepackage{enumitem}
\usepackage{mathabx}
\usepackage{tabularx}
\usepackage{mdframed}
\usepackage{longtable}
\usepackage{tabu}
\usepackage{nccmath}
\usepackage{cuted}

\newcommand{\B}{\{0,1\}}
\newcommand\true{\mbox{\sc True}}
\newcommand\false{\mbox{\sc False}}

\newcommand{\poly}{\mathsf{poly}}
\newcommand{\polyltone}{\mathsf{poly}_{<1}}

\newcommand{\Cut}{\mathsf{Cut}}
\newcommand{\Cov}{\mathsf{Cov}}
\newcommand{\Var}{\mathsf{Var}}

\newtheorem{lemma}{Lemma}
\newtheorem{theorem}{Theorem}
\newtheorem{claim}{Claim}
\newtheorem{fact}{Fact}
\newtheorem{remark}{Remark}
\newtheorem{corollary}{Corollary}
\newtheorem{definition}{Definition}

\newcommand{\defeq}{\stackrel{\textup{def}}{=}}
\newcommand{\eps}{\varepsilon}
\renewcommand{\epsilon}{\varepsilon}
\newcommand{\nfrac}{\nicefrac}
\newcommand{\E}{\mathop{\mathbf{E}}}

\newcommand{\activedeg}{\mathsf{actdeg}}
\newcommand{\activedist}{\mathsf{actdist}}

\newcommand{\cnt}{\mathsf{count}}

\newcommand{\Act}{\mathsf{Active}}
\newcommand{\varest}{\mathsf{Uvar}}
\newcommand{\meanest}{\mathsf{Lmean}}
\newcommand{\val}{\mathsf{val}}

\newcommand{\sdp}{\mathsf{SDP}}

\newcommand{\maxcut}{\mbox{\sc Max-Cut}\xspace}
\newcommand\maxbisection{\mbox{\sc Max-Bisection}\xspace}

\newcommand{\calP}{{\mathcal P}}

\newcommand{\calE}{{\mathcal E}}
\newcommand{\calL}{{\mathcal L}}

\newcommand{\R}{{\mathbf{R}}}

\newcommand{\flag}{\mathsf{flag}}

\renewcommand{\v}[1]{\mathbf{#1}}

\newcommand{\algowmaxcut}{{\sc Alg-Sim-MaxCUT}}

\newcommand{\dsay}[1]{}

\begin{document}

\title{\Large Near-optimal approximation algorithm for simultaneous $\maxcut$}
\author{Amey Bhangale
\thanks{Department of Computer Science, Weizmann Institute of Science.
Research supported in part by NSF grant CCF-1253886. This work was
done when the author was a graduate student at Rutgers University, New Brunswick, NJ, USA.}
\and
Subhash Khot\thanks{Department of Computer Science, Courant Institute of Mathematical Sciences, New York University. Supported by the NSF
Award CCF-1422159, the Simons Collaboration on Algorithms and Geometry and the Simons Investigator Award.}
\and
Swastik Kopparty \thanks{Department of Mathematics \& Department of
  Computer Science, Rutgers University.
  Research supported in part by a Sloan Fellowship, NSF grants CCF-1253886 and CCF-1540634, and BSF grant 2014359.}
  \and
Sushant Sachdeva\thanks{Department of Computer Science, University of Toronto. This work was
done when the author was a research scientist at Google, Mountain View,
CA, USA}
\and
Devanathan Thiruvenkatachari\thanks{Department of Computer Science, Courant Institute of Mathematical Sciences, New York University. Supported by same sources as Khot.}
}
\maketitle

\begin{abstract}
In the simultaneous $\maxcut$ problem, we are given $k$ weighted graphs on the same set of $n$ vertices, and the goal is to find a cut of the vertex set so that the minimum, over the $k$ graphs, of the cut value is as large as possible. Previous work \cite{BKS15} gave a polynomial time algorithm which achieved an approximation factor of $1/2 - o(1)$ for this problem (and an approximation factor of $1/2 + \epsilon_k$ in the unweighted case, where $\epsilon_k \rightarrow 0$ as $k \rightarrow \infty$).

In this work, we give a polynomial time approximation algorithm for simultaneous $\maxcut$ with an approximation factor of $0.8780$ (for all constant $k$). The natural SDP formulation for simultaneous $\maxcut$ was shown to have an integrality gap of $1/2+\epsilon_k$ in \cite{BKS15}. In achieving the better approximation guarantee, we use a stronger Sum-of-Squares hierarchy SDP relaxation and  a rounding algorithm based on Raghavendra-Tan \cite{RT12}, in addition to techniques from \cite{BKS15}.

\end{abstract}

\section{Introduction}

In this paper, we give near-optimal approximation algorithms
for the simultaneous \maxcut problem. Here we are given a collection
of weighted graphs $G_1, G_2, \ldots, G_k$ on the same
vertex set $V$ of size $n$. Our goal is to find a partition of the vertex set
$V$ into two parts, such that in {\em every} graph,
the total weight of edges going between the two parts is large.
The $k=1$ case is the classical \maxcut problem, and the approximability
of this problem has been
extensively studied~\cite{FL92,GW95,H01,KKMO07,MOO05,OW08}.
This paper studies the approximability of this problem for constant $k$.

We fix some convenient notation. Let the weighted graphs $G_1, \ldots, G_k$
be given by weight functions $\calE_1, \ldots, \calE_k$, which assign
to each pair in ${V \choose 2}$ a weight in $[0,1]$. We assume
that for each $i \in [k]$, the total weight of all edges under $\calE_i$ equals $1$. Let $f: V \to \{0,1\}$
be a function, which we view as a partition of the vertex set. We define $\val(f, \calE_i)$ to be the total weight (under $\calE_i$) of the edges
cut by the partition $f$. Given this setup, we can formally state the 
notions of approximation that we consider.
\begin{itemize}
\item {\bf $\alpha$-minimum approximation:} Let $c$ be the maximum, over all 
partitions $f^* : V \to \{0,1\}$, of the quantity $\min_{i \in [k]} \val(f^*, \calE_i)$. The goal is to output an $f: V \to \{0,1\}$ such that
$\min_{i \in [k]} \val(f, \calE_i) \geq \alpha \cdot c$.
\item {\bf $\alpha$-Pareto approximation:} Let $c_1, c_2, \ldots, c_k$
be given such that there exists $f^* : V \to \{0,1\}$ with
$\val(f^* , \calE_i) \geq c_i$ for each $i \in [k]$. The goal is to output
an $f: V \to \{0,1\}$ such that $\val(f, \calE_i) \geq \alpha \cdot c_i$ for all $i \in [k]$.
\end{itemize}

For $k = 1$, there is a celebrated polynomial time $\alpha_{GW} = 0.8786\ldots$ factor (Pareto) approximation algorithm
by Goemans and Williamson~\cite{GW95}. This approximation is in both the minimum and Pareto senses. 
Furthermore, it is Unique-Games hard to achieve a better approximation factor than this~\cite{KKMO07}, and
the entire polynomial time ``approximation curve" is also known.

For larger (but constant) $k$, far less is understood. Clearly, the hardness results from the $k=1$ case carry over, and 
thus it is Unique-Games hard to approximate this to a factor better than $\alpha_{GW}$.~\cite{ABG06} gave a polynomial time $0.439$-Pareto  approximation algorithm
for this problem for the case $k = 2$. Subsequently,~\cite{BKS15} gave a polynomial time $(1/2-\epsilon)$-Pareto approximation algorithm
for this problem. For the case of unweighted graphs\footnote{We call an instance of simultaneous \maxcut \ {\em unweighted} if for any $i$, all the nonzero weight edges under $\calE_i$ 
have the same weight.}, \cite{BKS15} showed that there is a polynomial time $(1/2 + \Omega(1/k^2))$-minimum approximation
algorithm. Furthermore, \cite{BKS15} gave a matching integrality gap of $(1/2 + O(1/k^2))$ for a natural SDP relaxation
of the minimum approximation problem.

Our main result is a polynomial time $0.8780$-factor Pareto approximation algorithm for simultaneous \maxcut for arbitrary constant $k$.

\medskip
\begin{theorem}
\label{theorem:main}
 For all constant $k$ and $c>0$, given weighted graphs $(G_i(V,\calE_i))_{i=1}^k$ with $|V|=n$ and where all non-zero edge weights are lower bounded by $exp(-n^c)$, there is a $\poly(n)$ time algorithm which computes a $0.8780$-factor Pareto approximation (and hence min approximation) to the 
simultaneous \maxcut problem with $k$ instances.
\end{theorem}
\medskip

\begin{remark}
\label{remark:edge weights}
We assume that the non-zero edge weights are lower bounded by $exp(-|V|^c)$ for some constant $c>0$. We are interested in an algorithm which runs in time polynomial in $|V|$ and hence it is natural to assume the edge weights are lower bounded by $exp(-|V|^c)$ as otherwise the bit complexity of the input will be super polynomial in $|V|$.
\end{remark}

\begin{remark}
Our approximation ratio matches the Goemans-Williamson constant $\alpha_{GW} = 0.8786\ldots$ up to three decimal places. It might be possible to improve the approximation ratio through small modifyications our rounding procedure. However, we believe that getting the exact $\alpha_{GW}$-approximation (if it exists) might require new techniques. See Remark~\ref{remark:roundingf} for more details.
\end{remark}

We give a brief overview of ideas involved in our algorithm next.
The main ingredients of the algorithm are: a sum-of-squares hierarchy SDP relaxation, a generalization of the ~\cite{RT12}, \cite{ABG12} approach to rounding
such relaxations, and some ideas from~\cite{BKS15}.

\subsection{Overview of the algorithm}

We begin by considering the unweighted case; later we will discuss how to remove this restriction. One crucial observation about the unweighted case is that if there are enough edges in every graph (as a function of $k$), then a random cut simultaneously cuts a constant fraction of edges from each graph with high probability. Thus, we can always assume that each target value is $c_i = \Omega_k(1)$, which is a constant for a constant $k$.

There is a natural SDP relaxation for the simultaneous \maxcut problem, generalizing the Goemans-Williamson SDP for the
$k = 1$ case. If we solve this SDP and round the resulting vector solution using the Goemans-Williamson hyperplane rounding procedure,
this gives us a distribution of partitions of the vertex set $V$, such that for {\em each} $i \in [k]$, the total
weight of edges cut in instance $i$ is at least $\alpha_{GW}$ times the corresponding SDP cut value. However, unlike 
in the $k=1$ case, this
does not guarantee the existence of a single partition of $V$ which is achieves a large cut value for all the $k$ instances simultaneously!
This distinction between distributions of solutions which are good in expectation for each instance and 
single solutions that are simultaneously good for all instances is at the heart of the difficulty in
designing simultaneous approximation algorithms.

One of the basic ingredients underlying mathematical programming relaxation hierarchies for combinatorial 
optimization problems is the idea of expanding the search space, from the discrete space of pure assignments
to the continuous space of distributions over assignments. For simultaneous approximation of \maxcut beyond a 
factor $1/2$, this idea alone is not enough. An example from~\cite{BKS15} shows that there are cases of simultaneous \maxcut on $k$-instances,
for which there is a distribution of partitions of $V$ cutting $(1 - \frac{1}{k})$-fraction of
edges {\em in expectation} for each instance, but for which any single partition of $V$, there is an instance 
$i \in [k]$, such that at most $1/2$ of the edges in instance $i$ are cut by the partition.
This is where the sum-of-squares SDP hierarchy comes in -- even though it is also modeled on the
idea of expanding the search space to distributions of assignments --  it allows us to {\em condition} on partial
assignments and impose a constraint that the SDP cut value is large in expectation for each instance and for every possible
conditioning on a small number of variables. This is what allows us to overcome the aforementioned obstacle.

Having formulated the SDP relaxation, we now discuss the rounding procedure.
The motivating observation is this: if the rounding procedure is such that for each
instance the expected cut value is large, and further the cut value is concentrated around its expectation
with high probability, then by a union bound, the rounding procedure will produce a cut
that is simultaneously good for all instances. The rounding procedure we will use will
be closely related to the Goemans-Williamson rounding (but different -- it was found by computer
search given various technical conditions required by the rest of the algorithm).
Our algorithm now tries to improve the concentration
of the cut-value produced by the rounding procedure, via a beautiful information-theoretic approach of Raghavendra-Tan~\cite{RT12}. 
If the cut-value for a certain instance turns out to be not concentrated under the rounding procedure, then it must be because
of high correlation between many pairs of edges of that instance (more precisely, correlation between the events that the edge
is cut). This in turn means that conditioning on the variables in a random edge should significantly decrease the
amount of entropy of the rounded cut. Iterating this several times, and using the fact that the initial entropy is not too large,
we conclude that conditioning on a small number of variables leads to good concentration for the rounding procedure.
The key point is that the sum-of-squares SDP relaxation we use gives us access to a vector solution for the
conditioned SDP, with the promise that the SDP cut-value (and hence the expected integral cut-value) is still large.
By the concentration property and a union bound, we get a simultaneously good cut.
This completes the description of the algorithm in the unweighted case.

To handle the general weighted case, we essentially need to overcome few technical obstacles.
Following~\cite{BKS15}, we add a preprocessing and postprocessing phase.
The preprocessing phase identifies ``wild" instances, i.e. those instances with an abnormally large number
of high (weighted-)degree vertices (which would increase the variance of the cut value of that instance
under random rounding).
Then the SDP based algorithm described above is run only on the ``tame" instances.

With conditioning on constantly many variables, we can only manage to bring the variance down to arbitrarily small constant.  Hence, in order to use second moment method to get concentration, we would need a good lower bound on the expected value of a cut given by our rounding procedure. If the graphs are weighted then it is not necessarily true that the simultaneous cut value is large for all instances.  One important property of the tame instances we used is that they have a good simultaneous $\maxcut$ value. We crucially use this property while formulating the SDP for tame instances. 

Finally in the postprocessing phase, we find suitable assignment to the high degree vertices
of the wild instances to ensure that those instance have a large cut value (without spoiling
the large cut value of the tame instances that the SDP guaranteed) -- this uses a new and much simpler 
perturbation argument compared to \cite{BKS15}.

This concludes the high-level description of the algorithm.

\subsection{Note about the rounding procedure}
\label{section:noterounding}
We mentioned earlier that our SDP solution after conditioning on a small number of variables is rounded by a rounding algorithm similar to the Goemans-Williamson rounding algorithm, but is different. We discuss this rounding procedure here, and compare it to the previous results that used similar rounding procedures.

For convenience, we switch the notation from $0/1$ to $+1,-1$, such that any function $f:V \rightarrow \{-1,+1\}$ defines a cut in a natural way. Define the bias of a $\{+1,-1\}$ random variable $x$ as $\E[x]$. The SDP solution induces a consistent local distribution on every set of variables of size at most some constant $r$, and we define the {\em SDP-bias} of a variable as the bias with respect to this local distribution. For a given rounding procedure, we define the {\em rounding-bias} of a variable as the bias with respect to the rounding procedure. Note that in the original hyperplane rounding of Goemans-Williamson, the rounding-bias of each vertex is 0. 

In the rounding procedure for the \maxbisection from \cite{RT12}, the rounding-bias for each variable induced by the rounding procedure is the same as the SDP-bias.
Their algorithm gave a 0.85 approximation for \maxbisection, and using the same bias function for the rounding along with the analysis of our algorithm, we can get 0.85 approximation for simultaneous \maxcut as well (See Section~\ref{section:round_analysis} for more details). The approximation factor given by \cite{RT12} was subsequently improved in \cite{ABG12} to 0.8776, where they used new techniques to relax the restriction on the choice of the bias function. Nevertheless, the rounding procedure was still quite constrained by the need to maintain the balance of the cut, as required by the \maxbisection problem.

In our setting, we do not need equal sized partition of the vertex set, we have more freedom in our rounding procedure with respect to the rounding-bias. It turns out that we only have to ensure that when the bias of a variable is high, the side of the cut it falls on is almost fixed (that this condition suffices heavily depends on features of our algorithm and its analysis). This helps us achieve an improved approximation factor of 0.8780. The rounding function we come up with was arrived at by computer search (along with some trial-and-error). 

The approximation ratio for our rounding procedure is proved by a computer assisted prover, using techniques similar to those of \cite{Sjo09} and \cite{ABG12}.

\subsection{Other related work}
The simultaneous \maxcut problem is a special case of the simultaneous approximation problem for general constraint
satisfaction problems. This general problem was studied in~\cite{BKS15}, where it was shown that there is
a polynomial time constant factor Pareto approximation algorithm for every simultaneous CSP (with approximation factor
independent of $k$). The algorithm there was based on understanding the structure of 
CSP instances whose value is highly concentrated under a random assignment to the variables, in addition
to linear-programming. It was also observed that there are CSPs for which the best polynomial time approximation factor for the simultaneous version
(with $k > 1$) is {\em different} from the best polynomial time approximation factor achievable in the standard $k =1$ case (assuming $P \neq NP$).
This makes  the study of simultaneous approximation factors very interesting.

The simultaneous MAXSAT problem was studied in~\cite{GRW11}, where a $1/2$-Pareto approximation algorithm was given.
For bounded width MAXSAT, the approximation factor was improved to $(3/4 - \epsilon)$ in~\cite{BKS15}.

It remains an open and very interesting problem to determine for which CSPs the simultaneous approximation problem for $k > 1$ is harder than the classical $k=1$ case.

\section{Preliminaries}

\subsection{Simultaneous $\maxcut$}
Let $V$ be a vertex set with $|V|=n$. We use the set $[n]$ for the vertex set $V$ for convenience. We are given $k$ graphs $G_1, \ldots, G_k$ on the vertex set $V$. Let $\calE_\ell : [n]\times [n] \rightarrow \R^{\geq 0}$ denote the edge weights of graph $G_\ell$ where the edge weights are normalized such that total weight of edges in each instance is $1$. As mentioned in Remark~\ref{remark:edge weights}, we assume that all edge weights are either $0$ or lower bounded by $2^{-n^c}$ for some $c>0$. We'll use $\calE_\ell$ to denote the edge set of graph $G_\ell$ and also the distribution of the edges based on the weights. For each instance $\ell$, we are given a target cut value $c_\ell$ that we would like to achieve (and we know is possible).

A partition $(U, \overline{U})$ of $V$ is said to be an $\alpha$-approximation if for each instance $G_\ell$,
$$\Cut_\ell(U, \overline{U}) \ge \alpha \cdot c_\ell.$$

\subsection{Information Theory}
In this section, we define and state some facts about entropy and mutual information between random variables.
\begin{definition}[Entropy]
Let X be a random variable taking values in $[q]$ then, entropy of $X$ is defined as:
$$H(X) := \sum_{i\in [q]}\Pr[X=i]\log \frac{1}{\Pr[X=i]}.$$
\end{definition}

\begin{definition}[Conditional Entropy]
Let X, Y be jointly distributed random variable taking values in $[q]$ then, the conditional entropy of $X$ conditioned on $Y$ is defined as:
$$H(X|Y) = E_{i\in [q]} H(X|Y=i).$$
\end{definition}

The following observations can be made about entropy of a collection of random variables. \\
Entropy of a collection of random variables cannot exceed the sum of their entropies.
\begin{fact}
\label{claim:entropy_sum}
$H(X_1, X_1,\ldots, X_n) \leq \sum_{i=1}^n H(X_i).$
\end{fact}

Entropy never decreases on adding more random variables to the collection.
\begin{fact}
\label{claim:entropy_setminus}
$H(X_1, X_2 | Y) \geq H(X_1 | Y).$
\end{fact}

Conditioning can only decrease the entropy.
\begin{fact}
\label{claim:entropy_cond}
$H(X|Y) - H(X|Y, Z) \geq 0.$
\end{fact}

\begin{definition}[Mutual Information]
Let X, Y be jointly distributed random variable taking values in $[q]$ then, the mutual information between $X$ and $Y$ is defined as:
$$ {I(X;Y) := \sum_{i,j\in [q]}\Pr[X=i, Y=j]\log \frac{\Pr[X=i, Y=j]}{\Pr[X=i]\Pr[Y=j]}}.$$
\end{definition}

\begin{theorem}(Data Processing Inequality)
\label{theorem:DPI}
If $X, Y, W, Z$ are random variables such that $X$ is fully-determined by $W$ and $Y$ is fully-determined by $Z$, then
$$I(X,Y) \leq I(W,Z).$$
\end{theorem}

\section{Algorithm for simultaneous weighted $\maxcut$}
In this section, we give our approximation algorithm for simultaneous weighted $\maxcut$ and the analysis.

\subsection{Notation}
We use the same notation as in \cite{BKS15}, which we reproduce here. Let $\calE = {V \choose 2}$ be the set of all possible edges.
Given an edge $e$ and a vertex $v$, we say $v \in e$ if $v$ appears in the edge $e$. For an edge $e$, let $e_1, e_2$ denote the endpoints of $e$ (arbitrary order).
Let $f: V \to \{0,1\}$ be an assignment. For an edge $e \in
\calE,$ define $e(f)$ to be 1 if the edge $e$ is cut by
the assignment $f$, and define $e(f) = 0$ otherwise. Note that an assignment cuts an edge if it assigns different values to the end points.
Then, we have the following expression for the cut value of the assignment:
$$\val(f,\calE) \defeq \sum_{e \in \calE} \calE(e) \cdot e(f).$$

A partial assignment $h:S\rightarrow\{0,1\}$ is an assignment to $S$ where $S \subseteq V$. We say an edge is {\it active} with respect to $S$ if at least one of the end vertices is not in $S$. We denote by $\Act(S)$ the set of all edges which are active with respect to $S$. For two edges $e,e' \in \calE,$ we say
$e \sim_{S} e'$ if they share a vertex that is contained in $V
\setminus S$. Note that if $e \sim_{S} e',$ then $e, e'$ are
both in $\Act(S)$, and also $e\sim_{S} e,$ $\forall e \in \calE$.
Let $\activedist_S(\ell)$ denote the distribution over $\Act(S)$, obtained by renormalizing $\calE_\ell$ to have total weight $1$ over $\Act(S)$.

Define the active degree given $S$ of a variable $v \in V\setminus S$ for instance $\ell$ by:
$$\activedeg_S(v, \ell) \defeq \sum_{ e \in \Act(S), e \owns v} \calE_\ell(e).$$
We then define the active degree of the whole instance $\ell$ given $S$:
$$\activedeg_S(\ell) \defeq \sum_{v \in V\setminus S} \activedeg_{S}(v,
\ell).$$
Note that we count weight of an active edge in $\activedeg_S(\ell)$ at most twice.  For a partial assignment $h: S \to \{0,1\}$, we define
$$\val(h,\calE_\ell) \defeq  \sum_{\substack{e \in \calE \\  e \notin \Act(S)}} \calE_\ell(e) \cdot e(h)$$
which is the total weight of non-active edges cut by the partial assignment $h$.
Thus, for an assignment $g: V\setminus S \to \{0,1\}$, to the
remaining set of variables, we have the equality:
$$\val(h\cup g,\calE_\ell) - \val(h,\calE_\ell) = \sum_{e \in \Act(S)}\calE_\ell(e) \cdot e(h\cup g).$$

\subsection{Algorithm}

In Figure~\ref{fig:weighted-maxcut}, we give the algorithm for Simultaneous \maxcut.
The input to the algorithm consists of an integer $k \ge 1,$ $\eps \in
(0,\nfrac{1}{5}],$ $k$ instances of \maxcut, specified by weight functions
$\calE_1,\ldots,\calE_k,$ and $k$ target objective values
$c_1,\ldots,c_k.$

\begin{figure*}[ht!]
\begin{tabularx}{\textwidth}{|X|}
\hline
\vspace{0mm}
{\bf Input}: $k$ instances of \maxcut , with weights defined by $\calE_1,\ldots,\calE_k$ on the
set of variables $V,$  target objective values $c_1,\ldots,c_k,$ and
$\eps \in (0,\nicefrac{1}{5}].$ \\
{\bf Output}: An assignment to $V.$ \\
{\bf Parameters:}  $\delta_0 = \frac{1}{10k}$, $\epsilon_0 = \frac{\eps}{2}$, $t = \frac{2k}{\gamma}\cdot\log\left( \frac{21}{\gamma}\right), \tau = \eps$, $\gamma=\frac{\tau^2\eps_0^2\delta_0}{4}$.\\
{\bf Pre-processing:}
\begin{enumerate}[itemsep=0mm, label=\arabic*.]
\item Initialize $S \leftarrow \emptyset$.
\item For each instance $\ell \in [k]$, initialize $\cnt_\ell
  \leftarrow 0$ and $\flag_\ell \leftarrow \true.$
\item Repeat the following until for every $\ell \in [k]$, either $\flag_\ell = \false$ or
$\cnt_\ell  = t$:
\label{item:alg:max-cut:loop}
\begin{enumerate}
\item For each $\ell \in [k]$, compute
$\varest_\ell = \sum_{e \sim_{S} e'} \calE_\ell(e) \calE_\ell(e').$
\item For each $\ell \in [k]$ compute
$\textstyle \meanest_\ell \defeq \tau \sum_{e \in \Act(S)} \calE_\ell(e).$
\item For each $\ell \in [k],$ if $\varest_\ell \geq \delta_0
  \epsilon_0^2 \cdot \meanest_\ell^2 $, then set $\flag_\ell = \true$,
  else set $\flag_\ell = \false$.
\item Choose any $\ell \in [k]$, such that
$\cnt_\ell < t$ AND $\flag_\ell = \true$ (if any):
\label{item:alg:max-cut:high-var}
\begin{enumerate}
\item
\label{item:alg:max-cut:high-deg}
Find $v \in V$ such that $\activedeg_S(v, \ell) \geq \gamma
\cdot \activedeg_S(\ell).$
\item Set $S \leftarrow S \cup \{v\}.$ We say that $v$ was brought
  into $S$ because of instance $\ell$.
\item Set $\cnt_\ell \leftarrow \cnt_\ell + 1$.
\end{enumerate}
\end{enumerate}
\item After exiting the loop:
\begin{itemize}
\item Let $\calL$ denote the set of all $\ell \in [k]$ for which $\flag_\ell$ is set to $\false$ (these will be called ``low-variance" instances).
\item Let $\mathcal H$ denote the set of all $\ell \in [k]$ for which  $\cnt_\ell = t$ (these will be called ``high-variance" instances).
\end{itemize}
\end{enumerate}
{\bf Main algorithm:}
\begin{enumerate}[resume]
\item
For each possible partial fixing $h:S\rightarrow \{0,1\}$ do the following
\begin{enumerate}
\item
Solve the SDP given in Figure~\ref{figure:newsdpsimmaxcut-lasserre} (Refer Section~\ref{sec:lasserre-formulation}).
\item
Follow the procedure in Figure~\ref{fig:indep} to make the solution locally independent. (Refer Section~\ref{sec:indep-local-sol})
\item
\label{item:alg:max-cut:rounding}
Round the solution based on the rounding procedure described in Figure~\ref{fig:rounding} to get a partial assignment $g:V\setminus S\rightarrow\{0,1\}$. (Refer Section~\ref{sec:rounding})
\item
\label{item:alg:max-cut:post-process}
{\bf Post-processing step: }
For every assignment $h':S\rightarrow \{0,1\}$, compute $\min\limits_\ell\frac{\val(h'\cup g, \calE_\ell)}{c_\ell}$ and return the assignment $h'\cup g$ that maximizes this.
\end{enumerate}
\end{enumerate}
\\
\hline
\end{tabularx}
\caption{Algorithm {\algowmaxcut} for approximating weighted simultaneous
  \maxcut}
  \label{fig:weighted-maxcut}
\end{figure*}

\subsection{Analysis of the Algorithm}

The algorithm broadly proceeds in 3 sections, the pre-processing step, the $\sdp$ step and the post processing step. The pre-processing step consists of identifying a small subset $S \subseteq V$ carefully. We then attempt all assignments to vertices in $S$ by brute force iteratively and use $\sdp$ with the partial assignment followed by a rounding to assign vertices in $V\setminus S$. The post-processing step involves perturbing the assignments to the vertices in $S$, the need for which is explained in detail in Section \ref{sec:postprocess}.

In what follows, we stick to the following notation.
Let $S^\star$ denote the final set $S$ that we get at the end of
Step~\ref{item:alg:max-cut:loop} of {\algowmaxcut}. Let $f^\star : V \rightarrow \{0,1\}$ be the assignment that achieves $\val(f^\star, \calE_\ell)\geq c_\ell$ for all $l\in [k]$  and $h^\star$ be the restriction of $f^\star$ to the set $S^\star$.

\subsubsection{Pre-processing: Low and High variance instances}
\label{sec:preprocess}

\dsay{We need a different notation for the smoothness parameter, $\alpha$ might be confusing given that we use it in the $\sdp$ for an assignment in $v_{\{S, \alpha\}}$.}
\begin{definition}
[$\tau$-smooth distribution]
A distribution $D$ on $\{0,1\}$ is called $\tau$-smooth if
$$\Pr_{x\sim D}[x=1]\geq \tau , \hspace{10pt}\Pr_{x\sim D}[x=0] \geq \tau.$$
\end{definition}

Let $h: S \to \{0,1\}$
be an arbitrary partial assignment to the vertices in $S$.  Let $g: V \setminus S \to \{0,1\}$ be the random assignment such that each of the marginals $g(v)$ is $\tau$-smooth. For an instance $\ell$, define the random variable
$$Y_\ell \defeq \val(h\cup g,\calE_\ell) - \val(h,\calE_\ell) = \sum_{e \in \Act(S)} \calE_\ell(e)\cdot e(h\cup g).$$
$Y_\ell$ measures the total active edge weight cut by the assignment in the instance $\ell$.

Consider the two quantities defined in Step~\ref{item:alg:max-cut:loop} of the algorithm. They depend only on $S$ (and importantly,
not on $h$), which will be useful in controlling the expectation and
variance of $Y_\ell$.  The first quantity is an upper bound on $\Var[Y_\ell]$:
$$\varest_\ell \defeq \sum_{e \sim_{S} e'} \calE_\ell(e) \calE_\ell(e').$$
The second quantity is a lower bound on $\E[Y_\ell]$:
$$\meanest_\ell \defeq \tau \cdot \sum_{e \in \Act(S)} \calE_\ell(e).$$
\begin{lemma}
\label{lemma:max-cut:uvar_lmean_relation}
Let $S \subseteq V$ be a subset of vertices and $h: S \to \{0,1\}$ be an arbitrary partial assignment to $S.$ Let $Y_\ell, \varest_\ell, \meanest_\ell$ be as above.
\begin{enumerate}
\item
\label{item:max-cut:uvar_lmean_relation:low_variance}
If $\varest_\ell \le \delta_0 \epsilon_0^2
  \cdot \meanest_\ell^2,$ then $\Pr[Y_\ell < (1-\epsilon_0) \E[Y_\ell] ] < \delta_0$.
\item If $\varest_\ell \ge \delta_0\epsilon_0^2 \cdot \meanest_\ell^2 $, then there exists $v \in V\setminus S$ such that
$$\activedeg_S(v, \ell) \geq \frac{1}{4}\tau^2\epsilon_0^2\delta_0 \cdot \activedeg_{S}(\ell).$$
\end{enumerate}
\end{lemma}
We defer the formal proof to the appendix.  The first part is a simple
application of the Chebyshev inequality. For the second part, we use
the assumption that $\varest_\ell$ is large, to deduce that there exists an edge $e$ such that the total weight of edges adjacent to the vertex/vertices in $e$ that belong to $V\setminus S$,  \emph{i.e.}, $\sum_{e_2 \sim_S
  e} \calE(e_2),$ is large. It then follows that at least one variable
$v \in e$ must have large active degree given $S.$

The above lemma (Lemma~\ref{lemma:max-cut:uvar_lmean_relation}) ensures that
Step~\ref{item:alg:max-cut:high-deg} in the algorithm always succeeds
in finding a variable $v$.  Next, we note that
Step~\ref{item:alg:max-cut:loop} always terminates. Indeed, whenever
we find an instance $\ell \in [k]$ in
Step~\ref{item:alg:max-cut:high-var} such that $\cnt_\ell < t$ and
$\flag_\ell = \true,$ we increment $\cnt_\ell.$ This can happen only
$tk$ times before the condition $\cnt_\ell < t$ fails for all $\ell
\in [k].$ Thus the loop must terminate within $tk$ iterations.

To analyze the
approximation guarantee of the algorithm, we classify instances
according to how many vertices were brought into $S^\star$ because of
them.
\begin{definition}[\small Low and High variance instances]
At the completion of Step~\ref{item:alg:max-cut:high-var} in Algorithm
{\algowmaxcut}, if $\ell \in [k]$ satisfies $\cnt_\ell = t$, we call
instance $\ell$ a {\em high variance} instance.  Otherwise we call
instance $\ell$ a {\em low variance} instance.
\end{definition}

The next two sections describes the SDPs that we formulate and solve for just the low variance instances. The claim that step~\ref{item:alg:max-cut:post-process} of the algorithm shown in Figure~\ref{fig:weighted-maxcut} handles the high variance instances is discussed and proved in Section~\ref{sec:postprocess}.

\subsubsection{Warmup: Basic $\sdp$ formulation for simultaneous $\maxcut$.}
\begin{figure*}[ht!]
\begin{tabularx}{\textwidth}{|X|}
\hline
{
\begin{align}
\sum_{e = \{i,j\} \in \calE_\ell} \calE_\ell(e)(\|\v{v_{\{(i,j), (0,1)\}}}\|_2^2 + \|\v{v_{\{(i,j), (1,0)\}}}\|_2^2)  \geq (1-3\epsilon)c_\ell  & \hspace{20pt} \forall \ell\in [k],\label{eq:objective}\\
\langle \v{v_{\{i,0\}}} , \v{v_{\{i,1\}}}\rangle = 0  &\hspace{20pt}  \forall i \in [n],\nonumber\\
\| \v{v_{\{(i,j), (b_1, b_2)\}}}\|^2 = \langle \v{v_{\{i,b_1\}}} , \v{v_{\{j,b_2\}}}\rangle & \hspace{20pt} \forall i,j\in [n] \nonumber\\
& \hspace{20pt}\text{ and } b_1, b_2\in \{0,1\}\nonumber\\
 \|\v{v_{\{T, \alpha\}}}\|^2 = \langle \v{v_{\{T, \alpha\}}}, \v{v_\emptyset}\rangle &\hspace{20pt} \forall T\subset V, |T|\leq 2, \alpha\in \{0,1\}^{|T|}\nonumber\\
 \v{v_{\{i, b\}}} = \v{v_\emptyset} &\hspace{20pt} \forall i \in S^\star, b = h(i)\nonumber\\
 \|\v{v_\emptyset}\|^2 = 1 &\hspace{20pt} \nonumber\\
 \sum_{e = \{i,j\} \in \Act(S^\star)} \hspace{-20pt}\calE_\ell(e)(\|\v{v_{\{(i,j), (0,1)\}}}\|_2^2 + \|\v{v_{\{(i,j), (1,0)\}}}\|_2^2) \geq \nfrac{\epsilon}{3}. \activedeg_{S^\star}(\ell) &\hspace{20pt} \forall \ell \in \mathcal{L}\label{eq:activedegreecut}
\end{align}
}
\tabularnewline \hline
\end{tabularx}
\caption{$\sdp^\star(h:S^\star\rightarrow\{0,1\})$ for simultaneous $\maxcut$ with partial fixing}
\label{figure:newsdpsimmaxcut}
\end{figure*}

Our algorithm involves formulating a {\em Lasserre Hierarchy} SDP relaxation of the
residual $\maxcut$ problem after giving a partial assignment $h: S^\star \to \{0,1\}$. In this section, as a warmup to its analysis, we present and study the {\em basic} version of that SDP.

We write the $\sdp^\star$ for simultaneous $\maxcut$ problem, after the partial fixing given by pre-processing step, as in Figure \ref{figure:newsdpsimmaxcut}. Let $\calL $ denote the set of indices of the low variance instances. We have vectors $\v{v_{T,\alpha}}$ for all $T$ and $\alpha$ where $T$ is a subset of $V$ of size at most 2, and $\alpha$ is an assignment to the vertices in $T$.

If we consider the $\sdp^\star$ without the constraint~(\ref{eq:activedegreecut}), it is easy to see that this is a relaxation. Given a partition $(U,\bar{U})$ of $V$ that achieves a simultaneous optimum, we can set vectors $\v{v_{T,\alpha}} = \v{v_\emptyset}$ if the pair $(T,\alpha)$ is consistent with $1_U$ (i.e. $1_U$ assigns $\alpha$ to $T$) and $\v{v_{T,\alpha}} = 0$ otherwise. $\v{v_\emptyset}$ can be viewed as a vector that denotes 1.

A part of our analysis require that for every low variance instance, the expected weighted fraction of active edges that we cut is at least a constant fraction of its active degree. An optimal SDP solution without constraint ~(\ref{eq:activedegreecut}) may not guarantee this condition (for the rounding procedure we choose). Hence, we force the SDP solution to satisfy this property by adding constraint ~(\ref{eq:activedegreecut}). We need to relax constraint ~(\ref{eq:objective}) to make sure that there is a solution that satisfies all the constraints.

We now prove that $\sdp^\star$, in its present form, has feasible solutions.
\begin{lemma}
\label{lemma:sdp_feasible}
 $\sdp^\star(h^\star)$ shown in Figure~\ref{figure:newsdpsimmaxcut} has a feasible solution.
\end{lemma}
\begin{proof}
To show that $\sdp^\star$ has a feasible solution, it suffices to show that there exists an integral solution that satisfies the constraints.

Fix an optimal assignment $f^\star : V \rightarrow \{0,1\}$ to the simultaneous instance. $f^\star$ satisfies $\forall \ell \in [k]$, $\val(f^\star, \calE_\ell) \geq c_\ell$. Consider the following random assignment: For all $v\in V\setminus S^\star$

\[r(v) =
\begin{cases}
      f^\star(v) & \mbox{with probability }   (1-\epsilon) \\
      \overline{f^\star(v)}  & \mbox{otherwise}\\
   \end{cases}
\]
where $\overline{f^\star(v)}$ is  $f^\star(v)$ flipped. For $v\in S^\star$, set $r(v) = f^\star(v)$. Now, for any $\ell\in \mathcal{L}$, let $Y_\ell$ denote the random variable
$$ Y_\ell = \sum_{e\in \Act({S^\star})} \calE_\ell(e) \cdot e(r).$$
We have $\E[e(r)] \geq \epsilon$, hence $\E[Y_\ell] \geq \nfrac{\epsilon}{2}\cdot\activedeg_{S^\star}(\ell)$. Also,

\begin{align*}
\E_r[\val(r, \calE_\ell) ]
& \geq  \sum_{e\notin \Act({S^\star})} \calE_\ell(e)\cdot\E[ e(r) ]+ \sum_{\substack{e\in \Act({S^\star}),\\ e(f^\star) = 1}} \calE_\ell(e) \cdot\E[e(r)]\\
& = \sum_{e\notin \Act({S^\star})} \calE_\ell(e)\cdot e(f^\star) +  \sum_{\substack{e\in \Act({S^\star}),\\ e(f^\star) = 1}} \calE_\ell(e) \cdot\min((1-\epsilon)^2 + \epsilon^2, 1-\epsilon)\\
& \geq (1-2\epsilon) \sum_{e: e(f^\star)=1} \calE_\ell(e) \\
& = (1-2\epsilon)\val(f^\star,\calE_\ell) \\
&\geq (1-2\epsilon) c_\ell.
\end{align*}

Thus, we have,
\begin{enumerate}
\item  $\E[Y_\ell] \geq \nfrac{\epsilon}{2}\cdot\activedeg_{S^\star}(\ell)$.
\item  $\E\limits_r[\val(r, \calE_\ell) ] \geq (1-2\epsilon) c_\ell.$
\end{enumerate}

Recall that the $\sdp^\star$ involves only the low variance instances. Also, the assignment $r$ is $\epsilon$-smooth on the set $V\setminus S^\star$. Therefore, we have concentration guarantees as given by point~\ref{item:max-cut:uvar_lmean_relation:low_variance} of Lemma~\ref{lemma:max-cut:uvar_lmean_relation}.

$$\Pr[Y_\ell \le (1-\epsilon_0)\E[Y_\ell]] \le \delta_0$$
$$\Pr[\val(r, \calE_\ell) \le (1-\epsilon_0)\E[\val(r, \calE_\ell)] ]\le \delta_0.$$
Hence, with probability at least $1-2\delta_0$, we have $Y_\ell \ge (1-\nfrac{\epsilon}{2})
\cdot\nfrac{\epsilon}{2}\cdot\activedeg_{S^\star}(\ell) \ge \nfrac{\epsilon}{3} \cdot\activedeg_{S^\star}(\ell)$ and $\val(r, \calE_\ell) \ge (1-\nfrac{\epsilon}{2})(1-2\epsilon)c_\ell \ge (1-3\epsilon)c_\ell$.

Now we do union bound over all low variance instances, we get with a probability at least $1-2\cdot \delta_0 \cdot k = \nfrac{4}{5}$, all the $\sdp$ constraints are satisfied by integral solution $r$. Thus, there exists an {\em integral} solution which satisfies all $\sdp^\star(h^\star)$ constraints and hence is feasible.
\end{proof}

\subsubsection{Lasserre Hierarchy $\sdp$ formulation.}
\label{sec:lasserre-formulation}

We now describe the $r^{th}$-level Lasserre SDP for the SDP in Figure~\ref{figure:newsdpsimmaxcut}.

\begin{figure*}[ht!]
\begin{tabularx}{\textwidth}{|X|}
\hline
{
\begin{align}
\sum_{e = \{i,j\} \in \calE_\ell} \left( \right. \calE_\ell(e)(\|\v{v_{\{S\cup\{i,j\}, \alpha \circ (0,1)\}}}\|_2^2 \hspace{20pt} & \hspace{20pt} \forall S\subseteq V, |S| \leq r-2, \alpha \in \{0,1\}^{|S|},\nonumber\\[-10pt]
+ \|\v{v_{\{S\cup\{i,j\}, \alpha \circ (1,0)\}}}\|_2^2)\left. \right)\hspace{20pt} & \hspace{40pt} \forall\ell\in [k]\nonumber\\[3pt]
\geq (1-3\epsilon)c_\ell \|\v{v_{\{S,\alpha\}}}\|^2\label{eq:lasserreobjective}\\[7pt]
 \sum_{e = \{i,j\} \in \Act(S^\star)}\left( \right. \calE_\ell(e)(\|\v{v_{\{S\cup\{i,j\}, \alpha \circ (0,1)\}}}\|_2^2 \hspace{20pt}&\hspace{20pt} \forall S\subseteq V, |S| \leq r-2, \alpha \in \{0,1\}^{|S|}, \nonumber \\[-10pt]
 + \|\v{v_{\{S \cup \{i,j\}, \alpha \circ (1,0)\}}}\|_2^2) \left. \right) \hspace{20pt}& \hspace{35pt}\forall \ell \in \mathcal{L }\nonumber\\[3pt]
\geq \nfrac{\epsilon}{3}. \activedeg_{S^\star}(\ell) \|\v{v_{\{S,\alpha\}}}\|^2  \label{eq:preserve_active}\\[7pt]
 \langle  \v{v_{\{S,\alpha\}}}, \v{v_{\{T,\beta\}}} \rangle = \|\v{v_{\{S\cup T,\alpha \circ \beta\}}}\|_2^2&\hspace{20pt} \forall S, T \subseteq V, |S\cup T| \leq r, \nonumber\\
 &\hspace{20pt} \alpha\in \{0,1\}^{|S|}, \beta \in \{0,1\}^{|T|}, \label{eq:consistency1}\\[7pt]
  \langle  \v{v_{S,\alpha}}, \v{v_{T,\beta}} \rangle = 0 &\hspace{20pt} \forall S, T \subseteq V, |S\cup T| \leq r, \alpha\in \{0,1\}^{|S|}, \beta \in \{0,1\}^{|T|},\nonumber\\[7pt]
 &\hspace{20pt} \text{s.t. } \alpha_{|S\cap T}\neq \beta_{|S\cap T}\label{eq:consistency2}\\[7pt]
 \|\v{v_{\{T, \alpha\}}}\|^2 = \langle \v{v_{\{T, \alpha\}}}, \v{v_\emptyset}\rangle &\hspace{20pt} \forall T\subseteq V, |T|\leq r, \alpha\in \{0,1\}^{|T|}\nonumber\\[7pt]
 \langle \v{v_{\{S, \alpha\}}}, \v{v_{\{i, b\}}} \rangle= \langle \v{v_{\{S, \alpha\}}}, \v{v_\emptyset}\rangle &\hspace{20pt} \forall S\subseteq V, |S| \leq r-1, \alpha \in \{0,1\}^{|S|}\nonumber\\[7pt]&\hspace{20pt} \forall i \in S^\star, b = h(i)\label{eq:partialfixing}\\
  \|\v{v_\emptyset}\|^2 = 1 &\hspace{20pt} \nonumber
\end{align}
}
\tabularnewline \hline
\end{tabularx}
\caption{$r$-round Lasserre lift of $\sdp^\star(h:S^\star\rightarrow\{0,1\})$ for simultaneous $\maxcut$ with partial fixing}
\label{figure:newsdpsimmaxcut-lasserre}
\end{figure*}

The $\sdp$ formulation has vectors $\v{v_{\{T,\alpha\}}}$ for all $T\subseteq V$ such that $|T| \leq r$ and $\alpha \in \{0,1\}^{|T|}$. In terms of local distribution, the $\sdp$ solution consists of {\em consistent local distribution} on every set $T$ of size at most $r$ (denoted by $\mu_T$). The random variable corresponding to set $T$ is denoted by $X_T$ distributed over $\{0,1\}^{|T|}$. The vector solution and the local distribution are related as follows: Suppose $T$ and $U$ are subsets of $V$ such that  $|T\cup U|\leq r$ and the assignments $\alpha \in \{0,1\}^{|T|}$ and $\beta \in \{0,1\}^{|U|}$ are consistent on $T\cap U$ then

$$ \langle \v{v_{T,\alpha}}, \v{v_{U,\beta}}  \rangle = \Pr_{\mu_{T\cup U}} (X_T = \alpha , X_U = \beta).$$

To ensure the consistency among local distributions, we have to add the constraints~\ref{eq:consistency1} and ~\ref{eq:consistency2} to the SDP in Figure~\ref{figure:newsdpsimmaxcut-lasserre}. Here if $\alpha \in \{0,1\}^{|S|}$ is an assignment to the vertices in $S$, and if $S'\subset S$, $\alpha_{| S'} \in \{0,1\}^{|S'|}$ denotes the assignment $\alpha$ restricted to the vertices in $S'$. Also, if $\alpha$ and $\beta$ are assignments to sets $S$ and $T$ agreeing on $S\cap T$, then we denote $\alpha \circ \beta$ an assignment to $S\cup T$. We also add the set of constraints (Equation~\ref{eq:partialfixing} in Figure~\ref{figure:newsdpsimmaxcut-lasserre}) to capture the partial assignment $h:S^\star\rightarrow\{0,1\}$ given by pre-processing.

With these definitions and constraints, the objective is to ensure that for all $\ell \in [k]$,
\begin{align*}
&\begin{aligned}[t]
 {\sum_{e=\{i,j\} \in \calE_\ell} \calE_\ell(e)\Pr}&  {\left[ X_{\{i,j\}} = (0,1) \vee X_{\{i,j\}} = (1,0) \right]}\\
& {\geq (1-3\epsilon)c_\ell}
\end{aligned}
\end{align*}

A simple way to capture this would be to write the objective of the SDP solution similar to the basic SDP formulation, as follows.

\begin{align*}
&\begin{aligned}[t]
 {\sum_{e = \{i,j\} \in \calE_\ell} \calE_\ell(e)}&  {\left(\|\v{v_{\{(i,j), (0,1)\}}}\|_2^2 + \|\v{v_{\{(i,j), (1,0)\}}}\|_2^2\right) }\\
& {\geq (1-3\epsilon)c_\ell}
\end{aligned}
\end{align*}

\begin{lemma}
\label{lemma:lasserre_feasible}
$r$-round Lasserre SDP shown in Figure~\ref{figure:newsdpsimmaxcut-lasserre} has a feasible solution.
\end{lemma}
\begin{proof}
Note that the feasible solution provided for the basic SDP in  Lemma~\ref{lemma:sdp_feasible} is integral. Therefore, we can directly conclude that the Lasserre lift of the SDP is feasible, as the same solution can be extended to the Lasserre SDP.

Assign $\v {v}_{S,\alpha}$ to $\v{v}_\emptyset$ if in the integral solution, the vertices in the set $S$ were assigned to $\alpha$ in that order, otherwise assign  $\v {v}_{S,\alpha}$ to $0$.
\end{proof}

In order to make the solution locally independent, we will need to condition based on the local distribution (Refer Section~\ref{sec:indep-local-sol}). Therefore, we need to re-write the objective so that it is satisfied (w.r.t the conditioned local distribution) even after conditioning on at most $r$ variables, as shown in Equation~\ref{eq:lasserreobjective} in the SDP formulation.

Also, similar to the previous case, we need to ensure that the solution post-conditioning still cuts at least a constant fraction of the active edges, which is ensured by adding the set of constraints specified in Equation~\ref{eq:preserve_active} in the SDP.

We observe that solving the SDP using ellipsoid method can result in a small additive error, and if $\activedeg_{S^\star}(\ell)$ is small compared to this additive error, the error would be significant. This will not cause any issues and we elaborate on this more. We can solve the SDP using ellipsoid method with an error of $\epsilon$ in time polynomial in $n$ and $\log(1/\epsilon)$. Therefore, we can take $\epsilon$ to be $\exp(-\poly(n))$ and still solve the SDP in time polynomial in $n$. We assumed that the non-zero edge weights are at least $\exp(-n^{c})$ for some constant $c>0$. Therefore, if the active degree is non-zero, it is at least $\exp(-n^{c})$. If we take $\epsilon  = \exp(-n^{c'})$ for $c'>>c$, we can solve the SDP in time polynomial in $n$ and get a vector solution which satisfies all the constraints upto additive error $\epsilon$ which is upto multiplicative factor of $(1+o(1))$. This will not have a major effect on our analysis and hence we assume from here onward that the vector solution that we get satisfies the all the constraints exactly.

\subsubsection{Obtaining independent local solution}
\label{sec:indep-local-sol}
The notion of independent solution (which is formalized below in Definition~\ref{def:local_independence}) that we need is different from \cite{RT12}. Following procedure in Figure~\ref{fig:indep} is used to achieve the kind of independence we need.

\begin{definition}
\label{def:local_independence}
A Lasserre solution is $\delta$-independent if it satisfies the following condition.
$$ \forall \ell \in \calL,  \E_{\substack{a, b\sim \activedist_{S^\star}(\ell) } }\left[\sum_{i,j \in \{1,2\}}I (X_{a_i};  X_{b_j})\right] \leq \delta.$$
\end{definition}

\begin{figure*}[ht!]
\begin{tabularx}{\textwidth}{|X|}
\hline
\vspace{5pt}
{\bf Input}: $r+2$ round Lasserre solution of a given simultaneous $\maxcut$ instance, $\delta \geq \frac{32k}{r}$\\
{\bf Output}: $\frac{\delta}{2}$-{\em independent} 2-round Lasserre solution.
\begin{enumerate}
\item For all $\ell_1, \ldots, \ell_{r/2} \in \calL$, and for all edges $e^i \in \activedist_{S^\star}(\ell_i)$ for all $i\in [r/2]$.
\begin{itemize}
\item Let $S = \cup_{i\in [r/2]} \{e^i_1, e^i_2\}$ be the endpoints of all the edges from $(1)$.
\item For every $\alpha \in \{0,1\}^{|S|}$ such that $\Pr[X_S = \alpha] >0$ in the local disctibution:
\begin{itemize}
\item Condition the $\sdp$ solution on the event $X_S = \alpha$.
\item Output if conditioned solution if it is $\frac{\delta}{2}$-independent.
\end{itemize}
\end{itemize}
\end{enumerate}
\tabularnewline \hline
\end{tabularx}
\caption{Making locally independent solution}
\label{fig:indep}
\end{figure*}

\begin{lemma}
\label{lemma:main_indep}
For all $\delta>0$, there exists $t\leq 2k/\delta$ and edges $e^1, e^2, \ldots, e^t\in \calE$ such that
\begin{align}
\label{eq:indep}
&\forall \ell\in \calL,\\
&\begin{aligned}[t] \E_{a,b \sim \activedist_{S^\star}(\ell)}  [I (X_{a_1},& X_{a_2};  X_{b_1}, X_{b_2} |\\
 &X_{e^1_1},  X_{e^1_2}, \ldots,  X_{e^t_1}, X_{e^t_2})]
\leq \delta\end{aligned}\nonumber
\end{align}
\end{lemma}
\begin{proof}

Consider the following potential function,
$$\phi = \sum_{\ell \in \mathcal{L}} \E_{a\in \activedist_{S^\star}(\ell)} H(X_{a_1}, X_{a_2}).$$
As entropy of a bit is at most $1$, clearly $\phi \leq 2k$. We have the following identity for each $\ell\in \calL$ which follows from conditional entropy and linearity of expectation

\begin{align*}
&\hspace{-10pt}\E_{a, b\in  \activedist_{S^\star}(\ell)} [ H(X_{a_1}, X_{a_2} | X_{b_1}, X_{b_2})] \\
&\hspace{10pt} \begin{aligned}[t] = \E_{a\in \activedist_{S^\star}(\ell)} &[ H(X_{a_1}, X_{a_2})] -  \\
\MoveEqLeft \E_{a,b\in  \activedist_{S^\star}(\ell)} I( X_{a_1}, X_{a_2} ; X_{b_1}, X_{b_2} )\end{aligned}
\end{align*}
This identity suggests that if for some $\ell\in \calL$,  $ \E_{a,b\in  \activedist_{S^\star}(\ell)} I( X_{a_1}, X_{a_2} ; X_{b_1}, X_{b_2} ) >\delta$ then there exists a conditioning which reduces the potential function by at least $\delta$. Thus, either the current conditioned solution satisfies (\ref{eq:indep}) in which case we are done or there exists an edge  $b$ such that if we condition the SDP solution based on the value of its endpoints $(b_1, b_2)$ according to the local distribution then the potential function decreases by at least $\delta$. So, if we fail to achieve (\ref{eq:indep}) then $\phi$ decreases by at least $\delta$. As entropy is always non-negative and conditioning never increases entropy (Fact~\ref{claim:entropy_cond}),  this process cannot go beyond $2k/\delta$ conditioning. Thus, before at most $2k/\delta$ conditioning, we are guaranteed to achieve (\ref{eq:indep}).

\end{proof}

The following fact follows from the data processing inequality (Theorem~\ref{theorem:DPI}).
\begin{fact}
\label{fact:mi_ub}
If $X_1, X_2, Y_1$ and $Y_2$ are random variables then for $i,j\in \{1,2\}$, we have
$$ I(X_i; Y_j) \leq I(X_1, X_2; Y_1, Y_2).$$
\end{fact}

The following corollary follows from Lemma~\ref{lemma:main_indep} and Fact~\ref{fact:mi_ub}.
\begin{corollary}
\label{corr:mi_reduction}
For all $\delta>0$, there exists $t\leq \frac{2k}{\delta}$, and edges $e^1, e^2, \ldots, e^t\in \calE$,  such that
\begin{align*}
\forall \ell\in \calL,\hspace{40pt}& \\
 {\E_{ a, b\sim \activedist_{S^\star}(\ell)} \Bigg[\sum_{i,j \in \{1,2\}}}& {I (X_{a_i};  X_{b_j} | }\\
& {X_{e^1_1},  X_{e^1_2}, \ldots,  X_{e^{t-1}_1}, X_{e^{t-1}_2})\Bigg] \leq 4 \delta}
\end{align*}
\end{corollary}

\begin{lemma}
\label{lemma:conditional_sol_prop}
There exists a fixing of at most $\frac{32k}{\delta}$ {\em variables} such that the conditioned solution is $\delta/2$ independent as well as  satisfies all constraints from $\sdp^\star(h^\star)$.
In particular, the algorithm in Figure~\ref{fig:indep} returns such a $\delta/2$ independent solution. Also, the running time is bounded by $n^{O(r)}$.
\end{lemma}
\begin{proof}
$\delta/2$ independence follows from Corollary~\ref{corr:mi_reduction} for $t = \frac{16t}{\delta}$ and Fact~\ref{fact:mi_ub}. Also, we can verify if a given $\sdp$ solution is $\delta/2$-independent or not in time polynomial in $n$. We now prove the later part.

As the conditioning maintains the marginal distribution of variables and because of the the Inequality (\ref{eq:lasserreobjective}) and (\ref{eq:preserve_active}), the constraints about the $\sdp$ cut value as well as the fraction of  active edges that are cut remain valid in the conditioned solution.  Hence,  from Lemma~\ref{lemma:sdp_feasible} $\sdp^\star(h^\star)$ remains feasible.
\end{proof}

\subsubsection{Rounding Procedure}
\label{sec:rounding}
In this section, we describe the rounding procedure for variables in $V\backslash S^\star$. The input to this procedure is 2 round Lasserre solution which is $\delta$-independent. We use a slight variation of GW rounding procedure to round the $\sdp$ vector solution. In particular, we want to maintain the bias of heavily biased random variable in our rounding procedure.\\

$\sdp$ gives the vector solution $\v{v_{i,0}}, \v{v_{i,1}}$ for all $i\in [n]$. Let $\mu_i = 2\E[X_i] - 1$, the expectation is according to the local distribution. Define $\v{v_i} = \v{v_{i,1}} - \v{v_{i,0}}$. These $\v{v_i}$ are the unit vectors (as $\|\v{v_i}\|^2 =  \|\v{v_{i,1}} - \v{v_{i,0}}\|^2 = \|\v{v_{i,1}}\|^2 + \|\v{v_{i,0}}\|^2 - 2\langle v_{i0}, v_{i1}\rangle = \Pr[X_i=0] + \Pr[X_i=1] - 0 = 1$). Let $\v{w_i}$ be component of $\v{v_i}$ orthogonal to $\v{v_\emptyset}$ ($\v{v_i} = \mu_i \v{v_\emptyset} + \v{w_i}$), $\| \v{w_i} \|_2 = \sqrt{1-\mu_i^2}$. Let $\overline{\v w}_i$ be the normalized unit vector of $\v{w_i}$. The rounding procedure is applied on vectors $\overline{\v w}_i$ along with the ``bias'' of each variable $\langle \v{v_i}, \v{v_\emptyset}\rangle$. The rounding procedure is shown in Figure~\ref{fig:rounding}.
\begin{figure*}[ht!]
\begin{tabularx}{\textwidth}{|X|}
\hline
\vspace{0pt}
{\bf Input}: $\delta$-independent $2$ round Lasserre solution, biases $\mu_i \in [-1,+1]$  and a function $f_R:[-1,1]\rightarrow[-1,1]$ which is bounded by above and below with some constant degree polynomials\\
{\bf Output}: A partition of $V$.
\begin{enumerate}
\item Pick a random Gaussian vector $\v{g}$ orthogonal to $\v{v_\emptyset}$ with each co-ordinate distributed as $\mathcal{N}(0,1)$.

\item For each $i\in [n]$
\begin{itemize}
\item Calculate $\xi_i = \langle \v g, \overline{\v w}_i \rangle$.

\item Let $r_i \leftarrow f_R(\mu_i)$

\item  Set $y_i = 1$ if $\xi_i \leq \Phi^{-1}(r_i/2 + 1/2) $, otherwise set $y_i = -1$. (Here, $\Phi$ is the Gaussian CDF)

\end{itemize}
\end{enumerate}
\tabularnewline \hline
\end{tabularx}
\caption{Rounding procedure}
\label{fig:rounding}
\end{figure*}

\subsubsection{Analysis of the rounding procedure}
\label{section:round_analysis}
 We use the notation $\polyltone(x)$ to denote a ``polynomial'' in $x$ with exponents as real numbers in $(0,1)$, such that $\polyltone(x)\rightarrow 0 $ as $x\rightarrow 0$.

Note that if we simply use the rounding function $f_R(x) = x$ as used in \cite{RT12} the we get for each instance,  in expectation the cut produced by the rounding procedure is at least $0.85$ times the SDP value (and hence eventually $0.85$ approximation for simultaneous $\maxcut$). Here, we leverage the fact that the constraints on what rounding functions are good for us are mild compared to \cite{RT12} as explained in Section~\ref{section:noterounding}.
\begin{lemma}
\label{lemma:approx_guarantee}
For a fixed low variance instance, the rounding procedure described in Figure~\ref{fig:rounding} gives an approximation ratio $0.878001(1-3\epsilon)$ in expectation for the following $f_R$,
\begin{align*}
f_R(x) &= 0.79\cdot x +0.07\cdot x^3 + 0.14\cdot x^7
\end{align*}
\end{lemma}
\begin{proof}
The proof of this lemma is numerical. We arrive at a informal approximate value for the bound using Matlab code (0.878001) and verify it using computer assisted techniques. The multiplicative loss of $(1-3\epsilon)$ is because of using $\sdp^\star$. We elaborate on the exact constant 0.878001 that we get next. The probability $p_{ij}$ that a given edge $(i,j)$ is cut by the rounding procedure is a function of $\mu_i$ and $\mu_j$, whereas its SDP contribution is  a quantity $q_{ij} := 1-\langle \v v_i, \v v_j\rangle/2$.  Thus to show a lower bound on approximation ratio it is sufficient to prove the same lower bound on $p_{i,j}/q_{ij}$ for all possible valid configurations of vectors. The program works in a recursive fashion, by continuously splitting the cube (all possible valid configuration) into sub-cubes. In each sub-cube, the program checks if either across all points in the region, the lower bound on $\alpha$ exceeds the approximation ratio we try to prove or if the upper bound on $\alpha$ is lower than the approximation ratio we try to prove. It proceeds with further division into smaller sub-cubes until one of the above is satisfied. If the latter is true at any point, the code returns a failure, and it returns a success if the entire region can be proved to come under the former case. The prover was adapted from \cite{ABG12} and modified to suit our rounding procedure. For more details on the workings of the prover, refer \cite{ABG12}.
\end{proof}

\begin{remark}
\label{remark:roundingf}
It seems possible to improve the constant 0.878001 by using a different $f_R$ which is continuous and satisfies $f_R(1) = 1$ and $f_R(-1) = -1$
However we suspect that a serious new idea would be needed to get a $\alpha_{GW}$-approximation algorithm.
\end{remark}

We need the following lemma from \cite{RT12}.

\begin{lemma}[\cite{RT12}]
\label{lemma:RT_vec_to_mi}
Let $\v{v_i}$ and $\v{v_j}$ be the unit vectors, $\v{w_i}$ and $\v{w_j}$ be the components
 of $\v{v_i}$ and $\v{v_j}$ that are orthogonal to $\v{v_\emptyset}$. Then $|\langle \v{w_i}, \v{w_j}\rangle| \leq 2I(X_i; X_j)$.
\end{lemma}

Above lemma along with Lemma~\ref{lemma:conditional_sol_prop} implies that if we sample edge $(i_1,i_2), (j_1,j_2)\sim  \activedist_{S^\star}(\ell)$ then we have on average,
$$   {|\langle \v{w_{i_1}} , \v{w_{j_1}} \rangle| +  |\langle \v{w_{i_1}} , \v{w_{j_2}} \rangle| +  |\langle \v{w_{i_2}} , \v{w_{j_1}} \rangle| +  |\langle \v{w_{i_2}} , \v{w_{j_2}} \rangle| \leq \delta. }$$
The rounding procedure is assigning values $\pm1$ to variables $y_i$ where $y_i$ is the variable for vertex $i\in V$ and its value decides on which side of cut the vertex $i$ is present in the final solution. Thus $y_i$ is a random variable taking values in $\{+1,-1\}$. We now wish to prove similar guarantee as the following lemma from \cite{RT12}, which relates the mutual information between the pair of rounded variables with the inner product of the corresponding vectors $w$.
\begin{lemma}[\cite{RT12}]
\label{lemma:RT_finalMI}
For $f_R$ such that $f_R(x) = x$, if $|\langle \v{w_i}, \v{w_j}\rangle| \leq \delta$ then $I(y_i; y_j) \leq \delta^{1/3}$.
\end{lemma}

In our case, we need that the mutual information between the events that a pair of edges are cut is small on average. Thus, our notion of local independence will be useful in proving this guarantee about mutual information.

\begin{lemma}
\label{lemma:low_mi_rounded}
Fix $f_R$ to be the rounding function given by Lemma~\ref{lemma:approx_guarantee}. For a pair of edges $(i_1, i_2)$ and $(j_1, j_2)$, suppose the vectors $w$ corresponding to their endpoints satisfy the following condition,
\begin{align*}
|\langle \v{w_{i_1}} , &\v{w_{j_1}} \rangle| +  |\langle \v{w_{i_1}} , \v{w_{j_2}} \rangle| + \\
&|\langle \v{w_{i_2}} , \v{w_{j_1}} \rangle| +  |\langle \v{w_{i_2}} , \v{w_{j_2}} \rangle| \leq \delta
\end{align*}
then $I(y_{i_1}y_{i_2}; y_{j_1}y_{j_2}) \leq \polyltone(\delta)$.
\end{lemma}
\begin{proof}
Since $\overline{\v w}_i$ is a normalized vector of $\v{w_i}$ and $\| \v{w_i}\| = \sqrt{1-\mu_i^2}$, we have
\begin{align}
\label{eq:local_corr}
\left.
\begin{array}{c}
\sqrt{1-\mu_{i_1}^2} \cdot \sqrt{1-\mu_{j_1}^2}\cdot|\langle \overline{\v w}_{i_1} , \overline{\v w}_{j_1} \rangle|\\
+ \sqrt{1-\mu_{i_1}^2} \cdot \sqrt{1-\mu_{j_2}^2}\cdot |\langle \overline{\v w}_{i_1} , \overline{\v w}_{j_2} \rangle|\\
+ \sqrt{1-\mu_{i_2}^2} \cdot \sqrt{1-\mu_{j_1}^2}\cdot|\langle \overline{\v w}_{i_2} , \overline{\v w}_{j_1} \rangle|\\
+ \sqrt{1-\mu_{i_2}^2} \cdot \sqrt{1-\mu_{j_2}^2}\cdot|\langle \overline{\v w}_{i_2} , \overline{\v w}_{j_2} \rangle|\\
\end{array} \right\} \leq \delta.
\end{align}

Since the total sum is bounded and each quantity is non-negative, at least one of the three quantities in each summand is at most $\delta^{1/3}$. We use two crucial properties of the rounding procedure:

\begin{itemize}
\item For the heavily biased variable according to the local distribution, the rounding procedure also keeps the rounded value heavily biased and
\item If two vectors $\v{w_i}$ and $\v{w_j}$ are nearly orthogonal, the corresponding rounded values $y_i$ and $y_j$ are nearly independent.
\end{itemize}

We need following claim which we prove in Section~\ref{section:proofclaim}.
\begin{claim}
\label{claim:orthogonal_vec_mi}
If all these quantities $|\langle \overline{\v w}_{i_1} , \overline{\v w}_{j_1} \rangle|, |\langle \overline{\v w}_{i_1}, \overline{\v w}_{j_2} \rangle|, |\langle \overline{\v w}_{i_2} , \overline{\v w}_{j_1} \rangle|, |\langle\overline{\v w}_{i_2} , \overline{\v w}_{j_2} \rangle|$ are upper bounded  by $\delta^{1/3}$, then we can upper bound $ I(  y_{i_1}, y_{i_2}) ; ( y_{j_1}, y_{j_2}) )\leq \polyltone(\delta)$.
\end{claim}

We now formally prove the upper bound on $I(y_{i_1}y_{i_2}; y_{j_1}y_{j_2})$ by case analysis.  We use the following upper bound which follows from data processing inequality.

$$I(y_{i_1}y_{i_2}; y_{j_1}y_{j_2}) \leq I(  y_{i_1}, y_{i_2}) ; ( y_{j_1}, y_{j_2}) ).$$
We now bound the right hand side based on following case analysis.

\begin{itemize}

\item Case 1: If all these quantities $|\langle \overline{\v w}_{i_1} , \overline{\v w}_{j_1} \rangle|, |\langle \overline{\v w}_{i_1}, \overline{\v w}_{j_2} \rangle|, |\langle \overline{\v w}_{i_2} , \overline{\v w}_{j_1} \rangle|, |\langle\overline{\v w}_{i_2} , \overline{\v w}_{j_2} \rangle|$ are upper bounded  by $\delta^{1/3}$ then using Claim~\ref{claim:orthogonal_vec_mi}, we can upper bound $ I(  y_{i_1}, y_{i_2}) ; ( y_{j_1}, y_{j_2}) )\leq \polyltone(\delta)$

\item Case 2: Consider the case when both the endpoints of an edge (w.l.o.g. of $(i_1, i_2)$) have large bias i.e. $\sqrt{1-\mu_{i_1}^2} \leq \delta^{1/3}, \sqrt{1-\mu_{i_2}^2} \leq \delta^{1/3}$.  It implies,

$$ \min( |1-\mu_{i_1}|, |1+\mu_{i_1}|) \leq \delta^{1/3}$$
$$\quad \min( |1-\mu_{i_2}|, |1+\mu_{i_2}|) \leq \delta^{1/3}$$

Assume both $\mu_{i_1}, \mu_{i_2}>0$ (there cases can be handled in a similar way). Then we have, $1-\mu_{i_1} \leq \delta^{1/3}$ and $1-\mu_{i_2} \leq \delta^{1/3}$.
Since the rounding procedure maintains the bias of a variable for a heavily biased variables, up to some constant polynomial factor,  we have,
\begin{align*}
I((  &y_{i_1}, y_{i_2}) ; ( y_{j_1}, y_{j_2}) )\\
&\leq H(y_{i_1}, y_{i_2}) \\
& \leq H(y_{i_1}) + H(y_{i_2}) \\
& = O(-(1 -\polyltone( \mu_{i_1})) \log(1 - \polyltone(\mu_{i_1})))  + \\
&\quad  \quad O(-(1 - \polyltone(\mu_{i_2})) \log(1 - \polyltone(\mu_{i_2})))\\
& \leq \polyltone(\delta).
\end{align*}

\item Case 3: Consider the case when exactly two non-endpoints of an edge (w.l.o.g. of $(i_1, j_i)$) have large bias. This implies that $\langle  \overline{\v w}_{i_2}, \overline{\v w}_{j_2}\rangle \leq \delta^{1/3}$.  Using the analysis of the previous case we have $H(y_{i_1}), H( y_{j_1})\leq \polyltone(\delta) $. Mutual information can be bounded as follows:
\begin{align}
I((&y_{i_1}, y_{i_2}); (y_{j_1}, y_{j_2}))\nonumber \\
& \leq H((y_{i_1}, y_{i_2})) - H((y_{i_1}, y_{i_2}) | (y_{j_1}, y_{j_2}))\nonumber \\
& \leq H(y_{i_1}) + H(y_{i_2}) - H(y_{i_2} | (y_{j_1}, y_{j_2})) \nonumber\\
& = H(y_{i_1})  + I((y_{j_1}, y_{j_2}) ; y_{i_2}) \label{eq:mi_ub}\\
& = \polyltone(\delta) + I((y_{j_1}, y_{j_2}) ; y_{i_2}).\label{eq:mi_onebiased}
\end{align}
Now,
\begin{align*}
I((&y_{j_1}, y_{j_2}) , y_{i_2}) \\
&= H((y_{j_1}, y_{j_2})) - H((y_{j_1}, y_{j_2}) | y_{i_2})\\
&\leq H(y_{j_1}) + H(y_{j_2}) - H( y_{j_2} | y_{i_2})\\
& =  H(y_{j_1}) + I(y_{j_2} ; y_{i_2}) \\
&= \polyltone(\delta) + I(y_{j_2} ; y_{i_2}).
\end{align*}
Therefore, we have
$$ I(y_{i_1}y_{i_2}; y_{j_1}y_{j_2}) \leq \polyltone(\delta) + I(y_{j_2} ; y_{i_2}).$$
From Claim~\ref{claim:orthogonal_vec_mi}, $I(y_{j_2} ; y_{i_2})$ is bounded above by $\polyltone(\delta)$ as $\langle  \overline{\v w}_{i_2}, \overline{\v w}_{j_2}\rangle \leq \delta^{1/3}$.

\item Case 4: Consider the only remaining case in which exactly one variable, say $X_{i_1}$, has a large bias i.e.  $\sqrt{1-\mu_{i_1}^2} \leq \delta^{1/3}$.
From (\ref{eq:local_corr}), it implies that pairwise inner products of $\overline{\v w}_{i_2},\overline{\v w}_{j_1}$ and $\overline{\v w}_{j_2}$  are at most $\delta^{1/3}$. Hence by Claim~\ref{claim:orthogonal_vec_mi}, we have $I(y_{i_2} ; (y_{j_1}, y_{j_2}) )\leq \polyltone(\delta)$. As before  from (\ref{eq:mi_ub}),
\begin{align*}
I(  y_{i_1}, y_{i_2}) ; ( y_{j_1}, y_{j_2}) )&\leq  H(y_{i_1})  + I((y_{j_1}, y_{j_2}) ; y_{i_2})  \\
&\leq \polyltone(\delta).
\end{align*}

\end{itemize}

\end{proof}
We can now upper bound the variance of a cut  produced by the randomized rounding in graph $\ell\in \calL$. Define $Y_\ell$ to be a random variable which is equal to the total weight of active edges cut by the rounding procedure.
$$ Y_\ell = \sum_{C \in \Act(S^\star)} \calE_\ell(C) e(g).$$

\begin{lemma}
\label{lemma:variance_final}
Fix a rounding function $f_R$ given in Lemma~\ref{lemma:approx_guarantee}  and let the $\sdp$ solution is $\delta$ independent then
$$\Var(Y_\ell) \leq \frac{\polyltone(\delta)}{\epsilon^2}\E[Y_\ell]^2.$$
\end{lemma}
\begin{proof}
Let $\alpha := 0.8780$. Note that by Lemma~\ref{lemma:approx_guarantee}, we have for an active edge $e(i,j)$,
\begin{equation}
\label{eq:edge_cut_guarantee}
\Pr[e(i,j) \mbox{ is cut }] \geq \alpha \cdot \frac{1-\langle \v{v_i}, \v{v_j}\rangle}{2}.
\end{equation}

We now lower bound the expected value of $Y_\ell$.
\begin{align*}
 {\E[Y_\ell]} & {=\sum_{e \in \Act(S^\star)} \calE_\ell(e) \cdot  \Pr[e(i,j) \mbox{ is cut}]} \\
(\mbox{ from } (\ref{eq:edge_cut_guarantee})) \hspace{-25pt}\\
&  {\geq \alpha \sum_{e \in \Act(S^\star)} \calE_\ell(e) \cdot   \frac{1-\langle \v{v_i}, \v{v_j}\rangle}{2}}\\
&  {=  \alpha \cdot\sum_{e \in \Act(S^\star)} \calE_\ell(e) (\|\v{v_{\{(i,j), (0,1)\}}}\|_2^2 + \|\v{v_{\{(i,j), (1,0)\}}}\|_2^2) }\\
(\mbox{ from } (\ref{eq:activedegreecut})) \hspace{-25pt}\\
&  {\geq \alpha \cdot \nfrac{\epsilon}{3}\cdot\activedeg_{S^\star}(\ell)}
\end{align*}
We can now bound the variance as follows:
\begin{align*}
 {\Var} {(Y_\ell)} &= {\sum_{i,j\in \Act(S^\star)} \calE_\ell(i)\calE_\ell(j)\Cov \left[\frac{1 - y_{i_1}y_{i_2}}{2}, \frac{1-y_{j_1}y_{j_2}}{2}\right]}\\
& =  {\sum_{i,j\in \Act(S^\star)} \calE_\ell(i)\calE_\ell(j)\left(\frac{1}{4}\cdot\Cov [y_{i_1}y_{i_2}, y_{j_1}y_{j_2}]\right)}\\
& \leq  {\sum_{i,j\in \Act(S^\star)} \calE_\ell(i)\calE_\ell(j) [O(\sqrt{I(y_{i_1}y_{i_2}; y_{j_1}y_{j_2})})]} \tag*{(from Lemma~\ref{lemma:low_mi_rounded})}) \\
& \leq  { \sum_{i,j\in \Act(S^\star)}\calE_\ell(i)\calE_\ell(j)\polyltone\left(\sum_{\substack{a\in \{{i_1},{i_2}\},\\ b\in \{{j_1},{j_2}\}}}|\langle \v{w_a} , \v{w_b} \rangle|\right) }\tag*{(from Lemma~\ref{lemma:RT_vec_to_mi})}\\
& \leq  {\sum_{i,j\in \Act(S^\star)} \calE_\ell(i)\calE_\ell(j)\polyltone\left(\E_{\substack{a\sim \{{i_1},{i_2}\},\\ b\sim \{{j_1},{j_2}\}}}[I(X_a;X_b)]\right)}\\
& \leq \begin{aligned}[t]\activedeg_{S^\star}&(\ell)^2{\times \E_{i,j\sim \activedist_{S^\star}(\ell)} \polyltone\E_{\substack{a\sim \{{i_1},{i_2}\}, \\b\sim \{{j_1},{j_2}\}}}\left[I(X_a;X_b)\right] }\end{aligned} \tag*{(from concavity of $\polyltone$)}\\
&\leq \begin{aligned}[t]\activedeg_{S^\star}&(\ell)^2  {\times\ \polyltone\left( \E_{\substack{({i_1},{i_2}),\\ ({j_1},{j_2}) }\sim \activedist_{S^\star}(\ell)} \E_{\substack{a\sim \{{i_1},{i_2}\}, \\b\sim \{{j_1},{j_2}\}}}\left[I(X_a;X_b)\right] \right)}\end{aligned}\\
&\leq \polyltone\left(\delta\right)\cdot \activedeg_{S^\star}(\ell)^2,
\end{align*}
Thus, we have
$$ \Var(Y_\ell) \leq \frac{\polyltone(\delta)}{\epsilon^2}\E[Y_\ell]^2.$$
\end{proof}

\begin{corollary}
\label{corollary:maxcut:lowvariance}
If we set $r:= \poly(k,\nfrac{1}{\eps})$ then for every low variance instance $\ell\in [k]$,with probability at least $1-1/10k$ we have $\val(h^\star \cup g)\geq (0.878001 -4\epsilon)c_\ell$.
\end{corollary}
\begin{proof}
Choosing $r$ a large constant (and thus $\delta$ very small), by Lemma~\ref{lemma:variance_final} and application of Chebyshev's Inequality, we can deduce that with probability
  at least $1-1/10k,$ we have $ Y_\ell \ge (1- \epsilon) \E[Y_\ell].$
  Thus, with probability at least $1-1/10k,$ we have,
\begin{align*}
  \val(h^\star\cup g,\calE_\ell) & = \val(h^\star,\calE_\ell) + Y_\ell \\
  & \ge \val(h^\star,\calE_\ell)  + (1- \epsilon) \E[Y_\ell]\\
  & \ge (1-\epsilon) \cdot \E[\val(h^\star,\calE_\ell) + Y_\ell] \\
  & = (1-\epsilon)\cdot \E[\val(h^\star\cup g,W_\ell)] \\
  & \ge (1-\epsilon)\cdot  0.878001\cdot (1-3\epsilon)\cdot c_\ell \\
  & \ge \left(0.878001 - 4\eps\right) \cdot c_\ell ,
\end{align*}
where we have used Lemma~\ref{lemma:approx_guarantee} for the lower bound $\E[\val(h^\star\cup g,W_\ell)] \geq
0.878001\cdot(1-3\epsilon)c_\ell,$
\end{proof}

\subsubsection{Post-Processing}

\begin{lemma}
\label{lemma:earlyvariables}
For all high variance instances $\ell\in [k],$ we have
\begin{enumerate}
\item $\activedeg_{S^\star}(\ell) \leq 2(1-\gamma)^t.$
\item For each of the first $\nicefrac{t}{2}$ variables that were brought
  inside $S^\star$ because of instance $\ell,$ the total weight of edges from $\calE_{\ell}$
  incident on each of that variable and totally contained inside $S^\star$ is at
  least $20\cdot\activedeg_{S^\star}(\ell).$
\end{enumerate}
\end{lemma}
\begin{proof}
  Consider any {\em high variance} instance $\ell \in[k]$. Initially,
  when $S=\emptyset,$ we have $\activedeg_\emptyset(\calE_\ell) \leq
  2$ since the weight of every edge is counted at most twice, once for
  each of the 2 active vertices of the edge,
  and $\sum_{e\in \calE} \calE_\ell(e)= 1$.
For every $v$, note that $\activedeg_{S_2}(v,\calE_\ell) \leq
  \activedeg_{S_1}(v,\calE_\ell)$ whenever $S_1 \subseteq S_2$.

  Let $u$ be one of the vertices that ends up in $S^\star$ because of
  instance $\ell.$ Let $S_u$ denote the set $S \subseteq S^\star$ just
  before $u$ was brought into $S^\star$.  When $u$ is added to $S_u$,
  we know that $\activedeg_{S_u}(u, \calE_\ell) \geq \gamma \cdot
  \activedeg_{S_u}(\ell).$ Hence, $\activedeg_{S_u \cup
    \{u\}}(\ell) \le \activedeg_{S_u}(\ell) -\activedeg_{S_u}(u,
  \calE_\ell) \le (1-\gamma) \cdot \activedeg_{S_u}(\ell).$ Since
  $t$ vertices were brought into $S^\star$ because of instance
  $\ell,$ and initially $\activedeg_{\emptyset}(\ell) \le 2,$ we
  get $\activedeg_{S^\star}(\ell) \le 2(1-\gamma)^t.$

  Now, let $u$ be one of the first $\nfrac{t}{2}$ vertices that ends
  up in $S^\star$ because of instance $\ell.$ Since at least
  $\nfrac{t}{2}$ vertices are brought into $S^\star$ because of
  instance $\ell,$ after $u,$ as above, we get
  $\activedeg_{S^\star}(\ell) \le (1-\gamma)^{\nfrac{t}{2}}
  \cdot \activedeg_{S_u}(\ell).$ Combining with
  $\activedeg_{S_u}(u, \calE_\ell) \geq \gamma \cdot
  \activedeg_{S_u}(\ell),$ we get $\activedeg_{S_u}(u,
  \calE_\ell) \geq \gamma (1-\gamma)^{-\nfrac{t}{2}}
  \activedeg_{S^\star}(\ell),$
  which is at least $21\cdot\activedeg_{S^\star}(\ell),$ by the
  choice of parameters. Since any edge incident on a vertex in
  $V \setminus S^\star$ contributes its weight to
  $\activedeg_{S^\star}(\ell),$ the total weight of edges
  incident on $u$ and totally contained inside $S^\star$ is at least
  $20\cdot\activedeg_{S^\star}(\ell)$ as required.
\end{proof}

\label{sec:postprocess}
We now describe a procedure {\sc Perturb} (see
Figure~\ref{fig:maxcut-perturb}) which takes $h^\star: S^\star \to \B$ and
$g : V \setminus S^\star \to \B$, and produces
a new $h : S^\star \to \B$ such that
for all (low variance as well as  high variance) instances $\ell \in [k]$,
$\val(h \cup g, \calE_\ell)$ is not much smaller than $\val(h^\star \cup g, \calE_\ell)$, and
furthermore, for all high variance instances $\ell \in [k]$, $\val(h
\cup g, \calE_\ell)$ is large. The procedure works by picking a
special vertex in $S^\star$ for every { high variance} instance and
perturbing the assignment of $h^\star$ to these special vertices. The
partial assignment $h$ is what we will be using to argue that
Step~\ref{item:alg:max-cut:post-process} of the algorithm produces a good
Pareto approximation.  More formally, we have the following Lemma.

\begin{figure*}
\begin{tabularx}{\textwidth}{|X|}
\hline
\vspace{0mm}
{\bf Input}: $h^\star: S^\star \to \B$ and
$g : V \setminus S^\star \to \B$\\
{\bf Output}: A perturbed assignment $h :  S^\star \to \{0,1\}.$
\begin{enumerate}[itemsep=0mm, label=\arabic*.]
\item Initialize $h \leftarrow h^\star.$
\item For $\ell = 1, \ldots, k$, if instance $\ell$ is a high variance
  instance case (i.e., $\cnt_\ell = t$), we pick a special variable
  $v_\ell \in S^\star$ associated to this instance as follows:
\begin{enumerate}
\item Let $B = \{ v \in V \mid \exists \ell \in [k] \mbox{ with }
  \sum_{e \in \calE, e \owns v} \calE_{\ell}(e)\cdot e(h \cup g) \geq
  \frac{\epsilon}{2k} \cdot\val(h\cup g, \calE_\ell) \}$.  Since the weight of
  each edge is counted at most twice, we know that $|B| \leq \frac{4k^2}{\epsilon}$.
\item Let $U$ be the set consisting of the first $t/2$ vertices
brought into $S^\star$ because  of instance $\ell$.
\item Since $\nfrac{t}{2} > |B|+ k$, there exists some $u \in U$
  such that $u \not\in B \cup \{v_1, \ldots, v_{\ell-1}\}$.  We define
  $v_\ell$ to be $u$.
\item By Lemma~\ref{lemma:earlyvariables}, the total $\calE_\ell$ weight of edges that are incident on $v_\ell$
  and only containing vertices from $S^\star$ is at least
  $20\cdot\activedeg_{S^\star}(\ell)$. We update $h$ by
setting $h(v_\ell)$ to be that value from $\{0,1\}$
such that at least half of the $\calE_\ell$ weight of these edges is
  satisfied.
\end{enumerate}
\item Return the assignment $h.$
\end{enumerate}
\\
\hline
\end{tabularx}
\caption{Procedure {\sc Perturb} for perturbing the optimal
  assignment}
  \label{fig:maxcut-perturb}
\end{figure*}

\begin{lemma}
\label{lemma:maxcut_perturbationeffect}
For the assignment $h$ obtained from Procedure {\sc Perturb} (see
Figure~\ref{fig:maxcut-perturb}), for each $\ell\in [k]$, $\val(h\cup
g,\calE_\ell) \geq (1-\nicefrac{\epsilon}{2}) \cdot \val(h^\star \cup
g, \calE_\ell)$. Furthermore, for each {\em high variance} instance
$\calE_\ell$, $\val(h\cup g,\calE_\ell) \geq
8 \cdot \activedeg_{S^\star}(\ell).$
\end{lemma}
\begin{proof}
  Consider the special vertex $v_\ell$ that we choose for {\em high
    variance} instance $\ell \in [k]$. Since $v_\ell \notin B,$ the
  edges incident on $v_\ell$ only contribute at most a
  $\nfrac{\eps}{2k}$ fraction of the objective value in each
  instance.
  Thus, changing the assignment $v_\ell$ can reduce the value of any
  instance by at most a $\frac{\epsilon}{2k}$ fraction of their
  current objective value. Also, we pick different special variables
  for each {\em high variance} instance. Hence, the total effect of
  these perturbations on any instance is that it reduces the objective
  value (given by $h^\star \cup g$) by at most $1 - (1-
  \frac{\epsilon}{2k})^k \le \frac{\epsilon}{2}$ fraction.  Hence for
  all instances $\ell \in [k]$, $\val(h\cup g,\calE_\ell) \geq
  (1-\nicefrac{\epsilon}{2})\cdot\val(h^\star\cup g,\calE_\ell)$.

  For a {\em high variance instance} $\ell \in [k]$, since $v_\ell \in
  U,$ the vertex $v_\ell$ must be one of the first $\nfrac{t}{2}$
  variables brought into $S^\star$ because of $\ell.$ Hence, by
  Lemma~\ref{lemma:earlyvariables} the total weight of edges
  that are incident on $v_\ell$ and entirely contained inside
  $S^\star$ is at least
  $20\cdot\activedeg_{S^\star}(\ell)$. Hence, there is an
  assignment to $v_\ell$ that satisfies at least at least half the
  weight of these \maxcut constraints in $\ell.$ At the end of the iteration when
  we pick an assignment to $v_\ell,$ we have $\val(h\cup g,
  \calE_\ell) \ge 10\cdot\activedeg_{S^\star}(\ell).$ Since the
  later perturbations do not affect value of this instance by more
  than $\nfrac{\epsilon}{2}$ fraction, we get that for the final
  assignment $h$, $\val(h\cup g,\calE_\ell) \ge
  (1-\nfrac{\epsilon}{2})\cdot 10 \cdot
  \activedeg_{S^\star}(\ell) \geq 8\cdot
  \activedeg_{S^\star}(\ell).$
\end{proof}

\begin{theorem}
\label{thm:max-cut}
Suppose we're given $\eps \in (0,\nfrac{1}{5}],$ $k$ simultaneous
\maxcut instances $\calE_1, \ldots, \calE_k$ on $n$ variables, and
target objective value $c_1, \ldots,c_k$ with the guarantee that there
exists an assignment $f^\star$ such that for each $\ell \in [k],$ we
have $\val( f^\star, \calE_\ell) \ge c_\ell.$ Then, the algorithm {\algowmaxcut} runs in time $\exp(\nfrac{k^3}{\eps^2}\log
(\nfrac{k}{\eps^2}))\cdot n^{\poly(k)},$ and with probability at least
$0.9,$ outputs an assignment $f$ such that for each $\ell \in [k],$ we
have, $\val(f, \calE_\ell) \ge \left( 0.878001 -5\epsilon\right)
\cdot c_\ell.$
\end{theorem}
\begin{proof}
Let $\alpha:=0.878001 $.
  By Corollary~\ref{corollary:maxcut:lowvariance} and a union bound, with
  probability at least $0.9$, over the choice of $g$,
  we have that for {\em every} low variance instance $\ell \in [k]$,
  $\val(h^\star \cup g, \calE_\ell) \geq (\alpha -
  4\epsilon) \cdot c_{\ell}$. Henceforth we assume that the
  assignment $g$ sampled in Step~\ref{item:alg:max-cut:rounding} of the
  algorithm is such that this event occurs. Let $h$ be the output of
  the procedure {\sc Perturb} given in Figure~\ref{fig:maxcut-perturb}
  for the input $h^\star$ and $g.$
  By Lemma~\ref{lemma:maxcut_perturbationeffect}, $h$ satisfies
\begin{enumerate}
\item\label{item:maxcut:proof:guarantees:1} For every instance $\ell \in [k]$, $\val(h\cup g,\calE_\ell) \geq (1-\nfrac{\eps}{2})\cdot \val(h^\star\cup g,\calE_\ell).$
\item \label{item:maxcut:proof:guarantees:2}For every high variance
  instance $\ell \in [k]$, $\val(h\cup g,\calE_\ell) \geq
 8 \cdot  \activedeg_{S^\star}(\ell).$
\end{enumerate}
We now show that the desired Pareto approximation behavior is achieved
when $h$ is considered as the partial assignment in
Step~\ref{item:alg:max-cut:post-process} of the algorithm. We analyze the
guarantee for low and high variance instances separately.

For any {\em low variance} instance $\ell \in [k],$ from
property~\ref{item:maxcut:proof:guarantees:1} above, we have
$\val(h\cup g,\calE_\ell) \geq (1-\nfrac{\eps}{2})\cdot
\val(h^\star\cup g,\calE_\ell)$. Since we know that $\val(h^\star \cup
g, \calE_\ell) \geq (\alpha-4\eps)\cdot c_\ell$, we
have $\val(h\cup g,\calE_\ell) \geq (\alpha-5\eps)\cdot c_\ell$.

For every high variance instance $\ell \in [k],$ since $h^\star =
f^\star|_{S^\star},$ for any $g$ we must have,
\begin{align*}
\val(h^\star \cup g, \calE_\ell) &\ge \val(f^\star, \calE_\ell) -
\activedeg_{S^\star}(\ell) \\
&\ge c_\ell - \activedeg_{S^\star}(\ell)
\end{align*}
Combining this with properties~\ref{item:maxcut:proof:guarantees:1}
and \ref{item:maxcut:proof:guarantees:2} above, we get,
\begin{align*}
\val(&h \cup g, \calE_\ell) \\
&\ge \left( 1 - \nfrac{\epsilon}{2}
\right) \cdot \max\{c_\ell - \activedeg_{S^\star}(\ell), 8 \cdot
\activedeg_{S^\star}(\ell)\} \\
&\ge \left( \alpha - \epsilon \right) \cdot c_\ell.
\end{align*}

Thus, for all instances $\ell \in [k]$, we get $\val(h \cup g) \ge
\left(\alpha-5\eps\right) \cdot c_{\ell}.$ Since we are taking
the best assignment $h\cup g$ at the end of the algorithm {\algowmaxcut}, the theorem follows.

\end{proof}
Plugging the appropriate value of $\epsilon$ in Theorem~\ref{thm:max-cut} completes the proof of $0.8780$-factor Pareto approximation (and hence min approximation) for simultaneous \maxcut for arbitrary constant $k$.

\section{Open Questions}

The main open question we would like to highlight is the question of determining optimal approximability and inapproximability results for simultaneous approximation of constraint satisfaction problems (CSPs).  In particular, it would be very interesting to develop techniques for showing nontrivial hardness of approximation in this context.

\section*{Acknowledgement}
We would like to thank the authors of \cite{ABG12} for making the prover code available for us. Our implementation of prover involves minor modifications of their code to suit our rounding algorithm. We also want to thank anonymous referees for helpful comments.
\bibliographystyle{alpha}
\bibliography{sim-maxcut}

\begin{thebibliography}{KKMO07}

\bibitem[ABG06]{ABG06}
Eric Angel, Evripidis Bampis, and Laurent Gourvès.
\newblock Approximation algorithms for the bi-criteria weighted max-cut
  problem.
\newblock {\em Discrete Applied Mathematics}, 154(12):1685 -- 1692, 2006.

\bibitem[ABG12]{ABG12}
Per Austrin, Siavosh Benabbas, and Konstantinos Georgiou.
\newblock Better balance by being biased: {A} 0.8776-approximation for max
  bisection.
\newblock {\em CoRR}, abs/1205.0458, 2012.

\bibitem[BKS15]{BKS15}
Amey Bhangale, Swastik Kopparty, and Sushant Sachdeva.
\newblock Simultaneous approximation of constraint satisfaction problems.
\newblock In {\em International Colloquium on Automata, Languages, and
  Programming}, pages 193--205. Springer, 2015.

\bibitem[FL92]{FL92}
Uriel Feige and L\'{a}szl\'{o} Lov\'{a}sz.
\newblock Two-prover one-round proof systems: Their power and their problems
  (extended abstract).
\newblock In {\em Proceedings of the Twenty-fourth Annual ACM Symposium on
  Theory of Computing}, STOC '92, pages 733--744, New York, NY, USA, 1992. ACM.

\bibitem[GRW11]{GRW11}
Christian Gla{\ss}er, Christian Reitwie{\ss}ner, and Maximilian Witek.
\newblock Applications of discrepancy theory in multiobjective approximation.
\newblock {\em CoRR}, abs/1107.0634, 2011.

\bibitem[GW95]{GW95}
Michel~X. Goemans and David~P. Williamson.
\newblock Improved approximation algorithms for maximum cut and satisfiability
  problems using semidefinite programming.
\newblock {\em J. ACM}, 42(6):1115--1145, November 1995.

\bibitem[H{\aa}s01]{H01}
Johan H{\aa}stad.
\newblock Some optimal inapproximability results.
\newblock {\em J. ACM}, 48(4):798--859, July 2001.

\bibitem[KKMO07]{KKMO07}
Subhash Khot, Guy Kindler, Elchanan Mossel, and Ryan O'Donnell.
\newblock Optimal inapproximability results for max-cut and other 2-variable
  csps?
\newblock {\em SIAM J. Comput.}, 37(1):319--357, April 2007.

\bibitem[MOO05]{MOO05}
Elchanan Mossel, Ryan O'Donnell, and Krzysztof Oleszkiewicz.
\newblock Noise stability of functions with low in.uences invariance and
  optimality.
\newblock In {\em Proceedings of the 46th Annual IEEE Symposium on Foundations
  of Computer Science}, FOCS '05, pages 21--30, Washington, DC, USA, 2005. IEEE
  Computer Society.

\bibitem[OW08]{OW08}
Ryan O'Donnell and Yi~Wu.
\newblock An optimal sdp algorithm for max-cut, and equally optimal long code
  tests.
\newblock In {\em Proceedings of the Fortieth Annual ACM Symposium on Theory of
  Computing}, STOC '08, pages 335--344, New York, NY, USA, 2008. ACM.

\bibitem[RT12]{RT12}
Prasad Raghavendra and Ning Tan.
\newblock Approximating csps with global cardinality constraints using sdp
  hierarchies.
\newblock In {\em Proceedings of the twenty-third annual ACM-SIAM symposium on
  Discrete Algorithms}, pages 373--387. SIAM, 2012.

\bibitem[Sjo09]{Sjo09}
Henrik Sjogren.
\newblock Rigorous analysis of approximation algorithms for max 2-csp.
\newblock 2009.

\end{thebibliography}

\appendix

\section{Deferred Proofs}
\label{section:proofclaim}

\subsection{Proof of Claim~\ref{claim:orthogonal_vec_mi}}
We need following bounds on the gaussian random variables.
\begin{claim}
\label{claim:gauss_large}
For all $x>0$, $\Pr_{g\sim \mathcal{N}(0,1)} [ |g|>x] \leq e^{\nfrac{-x^2}{2}}.$
\end{claim}

\begin{claim}
\label{claim:gauss_small}
For all $1>x>0$, $\Pr_{g\sim \mathcal{N}(0,1)} [ |g|<x] \leq x.$
\end{claim}

\paragraph{Random process $\calP$:} Let $\v w_1, \v w_2, \v w_3, \v w_4\in \R^4$ be unit vectors  and  $\mu_1, \mu_2, \mu_3, \mu_4$ be any  real numbers. Consider the following random variables $(y_1, y_2, y_3, y_4)$ where $y_i\in \{-1,+1\}$ which are sampled as follows: Pick a random vector $\v{g}:=(g_1, g_2, g_3, g_4)\in \R^4$ with each entry distributed as $\mathcal{N}(0,1)$.  Set
\begin{align*}
y_i &= -1 \hspace{20pt}\mbox{if $\langle \v g, \v w_i\rangle \leq \mu_i$,} \\
&= +1  \hspace{20pt}\mbox{otherwise.}
\end{align*}
The following lemmas gives sufficient conditions when $I(y_1, y_2 ; y_3, y_4)$ is {\em small}.

\begin{lemma}
\label{lemma:closeness}
Suppose $|\langle\v  w_i,\v w_j\rangle | \leq \delta$ for all $i, j\in [4]$, $i\neq j$ and $y_i s$ are sampled according to the random process $\calP$, then for all $\v b\in \{-1,+1\}^4$, we have
$$\left|\Pr[(y_1, y_2, y_3, y_4) = \v b] - \prod_{1\leq i\leq 4} \Pr[y_i = b_i] \right|=  O(\delta^{\nfrac{1}{4}}),$$
In fact, the joint distribution on any subset of variables is close to its product distribution {\em pointwise} with an additive error of at most $O(\delta^{\nfrac{1}{4}})$.
\end{lemma}
\begin{proof}
Assume that $0 < \delta < \nfrac{1}{100}$ (otherwise, the lemma is trivial). Let $\v e_i$ is a unit vector with $1$ in the $i^\text{th}$ coordinate. By rotational symmetry, we can assume that $\langle \v{w}_i, \v{e}_i\rangle \geq 1-20\delta$ for all $i$. We can write vector $\v w_i = \sqrt{1-\delta_i}\v e_i + \sqrt{\delta_i}\v \eta_i$ where $\v \eta_i$ is a unit vector orthogonal to $\v e_i$. The conditions on inner products therefore imply  each $\delta_i < 40\delta$. We will prove the lemma for $\v b  = (-1,-1,-1,-1)$ (all other cases are similar). We have,
\begin{align*}
\Pr[ y_i =-1, \forall i\in [4]] &= \Pr[ \forall i,  \langle \v g, \v w_i \rangle \leq \mu_i]\\
& = \Pr[ \forall i, \sqrt{1-\delta_i}g_i + \sqrt{\delta_i}\langle\v g , \v \eta_i\rangle\leq \mu_i]
\end{align*}
Let $B$ be the following event,

\noindent $B :$ There exists $1\leq i\leq 4$, such that $|\langle \v g, \v \eta_i\rangle| \geq \nfrac{1}{\delta^{\nfrac{1}{4}}}$.

By union bound,
\begin{align*}
\Pr[B] &= \sum_i \Pr[|\langle \v g, \v \eta_i\rangle| \geq \nfrac{1}{\delta^{\nfrac{1}{4}}}] \\
&\leq 4\cdot  \Pr[|\langle \v g, \v \eta_1\rangle| \geq \nfrac{1}{\delta^{\nfrac{1}{4}}}] \\
&= 4\cdot \Pr_{g\sim \mathcal{N}(0,1)}[|g| \geq \nfrac{1}{\delta^{\nfrac{1}{4}}}]\\
&\leq 4e^{-\frac{1}{2\sqrt{\delta}}},
\end{align*}
where last inequality uses Claim~\ref{claim:gauss_large}. Now,
\begin{align}
\Pr[ y_i =-1, \forall 1\leq i\in [4]] & = \Pr[B]\cdot \Pr[y_i =-1, \forall i\in [4] | B]+ \Pr[\overline{B}]\cdot \Pr[y_i =-1, \forall i\in [4] | \overline{B}] \nonumber\\
& \leq 4e^{-\frac{1}{2\sqrt{\delta}}} \cdot 1 + \Pr[y_i =-1, \forall i\in [4] | \overline{B}], \label{eq:close1}
\end{align}
 We now estimate the probability conditioned on event $\overline{B}$.
\begin{align}
 \Pr[y_i =-1, \forall i\in [4] | \overline{B}]  & =  \Pr[ \forall i, \sqrt{1-\delta_i} g_i + \sqrt{\delta_i}\langle\v g , \v \eta_i\rangle\leq \mu_i | \overline{B}]\nonumber\\
&\leq  \Pr[ \forall i, \sqrt{1-\delta_i}g_i \leq \mu_i + \sqrt{\delta_i} \cdot \frac{1}{\delta^{\nfrac{1}{4}}}]\tag*{($g_i$ independent)}\\
& =\prod_i \Pr[\sqrt{1-\delta_i} g_i \leq \mu_i + \sqrt{\delta_i} \cdot \frac{1}{\delta^{\nfrac{1}{4}}}]\tag*{(using $\delta_i \leq 40\delta$)}\\
&\leq \prod_i \Pr[\sqrt{1-\delta_i} g_i \leq \mu_i + \sqrt{40}\delta^{\nfrac{1}{4}} \tag*{(using $\delta_i \leq \nfrac{1}{2}$)}\\
&\leq \prod_i \Pr[g_i \leq (1+\delta_i)(\mu_i + \sqrt{40}\delta^{\nfrac{1}{4}} ) ]\tag*{(using $\delta_i \leq \nfrac{1}{2}$)}\\
&\leq \prod_i \Pr[g_i \leq \mu_i + \delta_i \mu_i + \nfrac{3}{2} \cdot  \sqrt{40}\delta^{\nfrac{1}{4}} ) ]\nonumber\\
&\leq \prod_i \Pr[g_i \leq (\mu_i +\delta_i \mu_i +  15\delta^{\nfrac{1}{4}} ) ].\nonumber
\end{align}
We now analyse the above probability in cases, and show the following:
\begin{align}
\Pr[g_i \leq \mu_i + \delta_i \mu_i + 15 \delta^{\nfrac{1}{4}}& )] \leq
\prod_i \Pr[g_i \leq \mu_i] + O(\delta^{\nfrac{1}{4}}) \label{eq:close2}
\end{align}
Notice that
\begin{align}
\prod_i \Pr[g_i \leq \mu_i + c\delta^{\nfrac{1}{4}} ] &\leq  \prod_i \Pr[ g_i \leq \mu_i]  +\Pr [ |g_i| \leq  c\delta^{\nfrac{1}{4}}]\nonumber\\
(\mbox{from Claim ~\ref{claim:gauss_small}})\hspace{5pt}& \leq \left( \prod_{1\leq i\leq 4} \Pr[y_i = b_i] + c\delta^{\nfrac{1}{4}}\right)\nonumber\\
& \leq  \prod_{1\leq i\leq 4} \Pr[y_i = b_i]  + O(\delta^{\nfrac{1}{4}}) \label{eq:close3}
\end{align}

\begin{itemize}
\item Case 1: $\mu_i < 0$.\\
In this case, we can directly say the following.
\begin{align*}
\prod_i \Pr[ g_i \leq \mu_i +& \delta_i \mu_i + 15\delta^{\nfrac{1}{4}})] \leq \prod_i \Pr[g_i \leq \mu_i + 15\delta^{\nfrac{1}{4}}].
\end{align*}
\item Case 2: $0\leq \mu_i \leq \frac{10}{\delta^{\nfrac{3}{4}}}$
We can say the following because $\delta_i < 40\delta$.
\begin{align*}
\prod_i \Pr[ g_i \leq \mu_i + \delta_i \mu_i &+ 15\delta^{\nfrac{1}{4}}] \leq \prod_i \Pr[g_i \leq \mu_i + O(\delta^{\nfrac{1}{4}}) ]
\end{align*}
\item Case 3: $\mu_i > \frac{10}{\delta^{3/4}}$
In this case, since $\mu_i$ is large, we have the following from Claim~\ref{claim:gauss_large}.

\[\prod_i \Pr[g_i \leq \mu_i] \geq 1 - o(\delta^{\nfrac{1}{4}})\]
Therefore,
\begin{align*}
\prod_i \Pr[g_i \leq \mu_i+ \delta_i \mu_i + 15\delta^{\nfrac{1}{4}}] \leq 1\leq \prod_i \Pr[g_i \leq \mu_i] + o(\delta^{\nfrac{1}{4}})
\end{align*}
\end{itemize}
Form (\ref{eq:close1}), (\ref{eq:close2}) and (\ref{eq:close3}) we get
$$\Pr[(y_1, y_2, y_3, y_4) = \v b] - \prod_{1\leq i\leq 4} \Pr[y_i = b_i]  \leq  O(\delta^{\nfrac{1}{4}}).$$
The other direction can be shown in an analogous way.
\end{proof}

 We can now bound the Mutual information between $(y_1, y_2)$ and $(y_3, y_4)$  if the vectors $\v w_i$ satisfy the condition from Lemma~\ref{lemma:closeness}

\begin{lemma}
\label{lemma:mi_small_suff_cond}
Suppose $|\langle\v  w_i,\v w_j\rangle | \leq \delta$ for all $i, j\in [4]$ and $i\neq j$, then $I( (y_1, y_2) ; (y_3, y_4) ) \leq \polyltone(\delta)$, where $y_i$ are sampled according to the random process $\calP$.
\end{lemma}
\begin{proof}
The lemma follows from Lemma~\ref{lemma:closeness} as the distribution is {\em close} to the product distribution.

To formally prove the lemma, first we assume that each of the random variables $y_i$ is not heavily biased i.e.   $\Pr[y_i = -1] \in [\delta^{1/100}, 1-\delta^{1/100}]$.
Using the definition of mutual information,

\begin{align}
I(& (y_1, y_2) ; (y_3, y_4) ) =\hspace{-10pt}\sum_{\substack{b_1, b_2, b_3, b_4 \\ \{-1+1\}}} \Bigg[[\Pr[\v y = \v b]\cdot \log \frac{ \Pr[\v y = \v b]}{ \Pr[ (y_1, y_2) = (b_1, b_2)] \cdot \Pr[ (y_3, y_4) = (b_3, b_4)]}\Bigg] \label{eq:mi_expansion}
\end{align}
Form Lemma~\ref{lemma:closeness}, we have
\begin{align*}
\Pr[ (y_1, y_2) = (b_1, b_2)]  &\geq \Pr[y_1=b_1]\Pr[y_2=b_2] -O(\delta^{\nfrac{1}{4}})\\
\Pr[ (y_3, y_4) = (b_3, b_4)]  &\geq \Pr[y_3=b_3]\Pr[y_4=b_4] -O(\delta^{\nfrac{1}{4}})
\end{align*}
Plugging any simplifying in (\ref{eq:mi_expansion}), we get
\begin{align*}
I( (y_1, y_2) ; (y_3, y_4) ) \leq  {\sum_{b_1, b_2, b_3, b_4 \{-1+1\}} \Pr[\v y = \v b]\cdot \log \frac{ \prod_{1\leq i\leq 4} \Pr[y_i = b_i] +O(\delta^{\nfrac{1}{4}})}{ \prod_{1\leq i\leq 4} \Pr[y_i = b_i] -O(\delta^{\nfrac{1}{4}})}}
\end{align*}
As each variable is not heavily biased, we have $\prod_{1\leq i\leq 4} \Pr[y_i = b_i] \geq \delta^{1/25}$ and hence the log in the above expression can be upper bounded by $\log \frac{ \delta^{1/25} + O(\delta^{\nfrac{1}{4}})}{ \delta^{1/25}- O(\delta^{\nfrac{1}{4}})}$ which is at most $\log (1+O(\delta^{1/10})) \leq O(\delta^{1/10})$. Hence we have
$$ I( (y_1, y_2) ; (y_3, y_4) ) \leq  O(\delta^{1/10}).$$
If a variable is heavily biased, suppose say $y_1$ has large bias, then we can claim $I( (y_1, y_2) ; (y_3, y_4) ) \leq \polyltone(\delta) +I(  y_2 ; (y_3, y_4) ) $ using derivation similar to (~\ref{eq:mi_onebiased}) and then proceed by upper bounding $I(  y_2 ; (y_3, y_4) )$ in a similar fashion as above.
\end{proof}

\noindent {\bf Proof of Claim~\ref{claim:orthogonal_vec_mi}:} The proof follows from Lemma~\ref{lemma:mi_small_suff_cond} noting the fact that the upper bound is independent of $\mu_i$.
\subsection{Proof of Lemma~\ref{lemma:max-cut:uvar_lmean_relation}}
\begin{proof}
Item 1 of the lemma follows from Chebyshev's inequality. We now focus on the proof of Item 2.
We have
\begin{align*}
\varest_\ell &\ge \delta_0\epsilon_0^2 \cdot \meanest_\ell^2\\
\Rightarrow \sum_{e \sim_{S} e'} \calE_\ell(e) \calE_\ell(e') &\ge  \delta_0\epsilon_0^2 \cdot \meanest_\ell^2
\end{align*}
Let $e_0$ be an edge in $\Act(S)$ that maximizes $\sum_{e \sim_{S} e_0} \calE_\ell(e)$. We can now upper bound the expression on the left as follows
$$\sum_{e \sim_{S} e'} \calE_\ell(e) \calE_\ell(e') \leq \sum_{e \sim_{S} e_0} \calE_\ell(e) \cdot \sum_{e\in\Act(S)} \calE_\ell(e).$$
Therefore, we have
\begin{align*}
\sum_{e \sim_{S} e_0} \calE_\ell(e) \cdot \sum_{e\in\Act(S)} \calE_\ell(e) &\ge  \delta_0\epsilon_0^2 \cdot \meanest_\ell^2 \\
&\ge  \delta_0\epsilon_0^2 \cdot \tau^2 \cdot \left(\sum_{e \in \Act(S)} \calE_\ell(e)\right)^2\\
\Rightarrow \sum_{e \sim_{S} e_0} \calE_\ell(e) &\ge  \delta_0\epsilon_0^2 \cdot \tau^2 \cdot \sum_{e \in \Act(S)} \calE_\ell(e)
\end{align*}
Let $v$ be the end vertex of $e_0$ that has greater weight of active edges adjacent to it, $v \in V\setminus S$.
We can say the following
$$\activedeg_S(v, \ell) \ge \frac{1}{2} \cdot \delta_0\epsilon_0^2 \cdot \tau^2 \cdot \sum_{e \in \Act(S)} \calE_\ell(e).$$
From the definition of $\activedeg_S(\ell)$, we can say the following
$$\activedeg_S(\ell) \le 2 \cdot \sum_{e \in \Act(S)} \calE_\ell(e),$$
as each edge could contribute at most twice to the sum, once for each end vertex. This gives us the following required result.
$$\activedeg_S(v, \ell) \ge \frac{1}{4} \cdot \delta_0\epsilon_0^2 \cdot \tau^2 \cdot \activedeg_S(\ell).$$
\end{proof}

\end{document}